\documentclass{article}
\usepackage[english]{babel}
\usepackage{amsthm} 
\usepackage{amssymb}
\usepackage{amsmath}
\newcommand{\lv}{\left\vert}
\newcommand{\rv}{\right\vert}
\newcommand{\esp}{^{\gamma}}

\newcommand{\espp}{^{(\lambda,\gamma)}}
\newcommand{\espl}{^{(\lambda)}}
\newcommand{\esppp}{^{(\lambda,\delta,\gamma)}}
\newcommand{\1}{_{(1)}}

\newcommand{\en}{_{(n)}}
\newcommand{\R}{\mathbb{R}}
\newcommand{\K}{\mathbb{K}}
\newcommand{\A}{\mathcal{A}}
\newcommand{\G}{\mathcal{G}}
\newcommand{\C}{\mathcal{C}}
\newcommand{\xin}{\varkappa}
\newcommand{\etan}{\varsigma}
\newtheorem{remark}{Remark}[section]
\newtheorem{note}{Note}[section]
\newtheorem{lemma}{Lemma}
\usepackage{varioref}
\usepackage{latexsym}

\begin{document}

\title{Long time asymptotics of a Brownian particle coupled with a random environment with non-diffusive feedback force}
\author{Michela Ottobre\\
email address:  m.ottobre08@imperial.ac.uk\\
postal address: Department of Mathematics\\ 
Room 617   Huxley Building\\ 
   Imperial College London \\
   180 Queens Gate 
   London SW7 2BZ 
   }

\hyphenation{sub-diffusive} 
\hyphenation{super-diffusion} 
\maketitle

\begin{abstract}
We study the long time behavior of a Brownian particle moving in    an anomalously diffusing field, the evolution of which depends on the particle position. We prove that the process describing the asymptotic behaviour of the Brownian particle has bounded (in time) variance when the particle  interacts with a subdiffusive field; when the interaction is with a superdiffusive field the variance of the limiting process grows in time as $t^{2\gamma-1}$, $1/2<\gamma<1$. Two different kinds of   superdiffusing (random) environments are considered: one is described through the use of  the fractional Laplacian; the other   via the  Riemann-Liouville fractional integral. The subdiffusive field is modeled through the Riemann-Liouville fractional derivative.   
\end{abstract}
\section{Introduction} 
In \cite{articolo}, L. Bertini et al considered the following system   of It\^o-SDEs, describing the evolution of a one-dimensional interface:
\begin{equation}\label{sistema_originale_diffusivo}
\left\{
\begin{array}{l}
dX(t)=\lambda dw(t)+ \alpha \langle\varphi_{X(t)},h(t) \rangle dt\\
\\
dh(t)=\frac{1}{2}\Delta h(t)dt - \varphi_{X(t)} dX(t),
\end{array}
\right.
\end{equation}
with initial conditions $X(0)=h(0)=0$.
In the above system $w(t)$ is a one dimensional  Brownian motion (BM) on the filtered probability space $(\Omega, \mathcal{F}, \mathcal{F}_t, P)$ 
($E$ is going to denote expectation with respect to $P$) and  
$\langle\cdot,\cdot\rangle$ is the scalar product of $L^2(\mathbb{R},dx)$.
More precisely, in \cite{articolo} the authors consider a system thermally isolated from the exterior, in a state in which two phases coexist. Under the assumption of planar symmetry for the system, the interface position is represented by the point $X(t)\in \mathcal{C}(\R_{+})$ separating the two phases. In equation (\ref{sistema_originale_diffusivo})$_1$ the interface displacements are described as the sum of two contributions: the first is a Brownian fluctuation, related to the macroscopic fluctuations of the system, the second is the interaction with a diffusive field, $h(t)=h(t,x)\in 
\mathcal{C}(\R_{+};\mathcal{C}(\R))$. Also, 
\begin{equation*}
\langle\varphi_{X(t)},h(t) \rangle=\int_{\R}dx\varphi(x-X(t))h(t,x),
\end{equation*}
where $\varphi(x)$ is a probability density in the Schwartz class  (regions of the field far from the interface do not significantly affect the interface evolution) and $\varphi_{X(t)}=\varphi(x-X(t))$.\\
On the other hand, equation (\ref{sistema_originale_diffusivo})$_2$ describes the field variation as the sum of a diffusive term plus a "feedback term" taking into account the latent heat effect.\\
The parameters $\lambda>0$ and $\alpha>0$  determine  the intensity of the Brownian noise and of the coupling with the field, respectively.
In \cite{articolo} the authors study a scaling limit of $X(t)$ as $\lambda\rightarrow0$ under the hypothesis $\alpha=\lambda$ of weak coupling.\\
Notice that the system (\ref{sistema_originale_diffusivo}) can also be interpreted as describing a Brownian motion weakly coupled with a (diffusive) random environment, the evolution of which depends on the position of the Brownian motion itself.
For further details about the model we refer to \cite{articolo} and \cite{articolo_2}.

Let $\xi(t)$ be the solution of the following integral equation
\begin{equation}\label{eqn_integrale_xi_diffusivo}
\xi(t)=\bar{b}(t)-\int_o^t ds\rho_{t-s}(0)\xi(s),
\end{equation}
where $\bar{b}(t)$ is the scaled BM  $\bar{b}(t)=\lambda w(t\lambda^{-2})$ and $\rho_t(x)=\rho(t,x)$ is the density of a centered Gaussian with variance $t$. In \cite{articolo} the  following asymptotics (\ref{risultato_diffusivo}) and (\ref{asintotica_di_xi_diffusivo})
are obtained: upon rescaling the interface position, i.e. considering the process
$X_{\lambda}(t)=X(t\lambda^{-2})$, we have  that $\forall N\in [1,\infty)\, \exists \,\tau=\tau(N)>0$ s.t.
\begin{equation}\label{risultato_diffusivo}
\lim_{\lambda\rightarrow 0} E\sup_{t\leq\tau\mid\log\lambda\mid}
\left\vert X_{\lambda}(t)-\xi(t)\right\vert^N=0\mbox{.}
\end{equation}
As noticed in \cite{articolo}, this implies that $X_{\lambda}$  converges weakly  to $\xi$ as $\lambda\rightarrow0$ in $\mathcal{C}(\R_{+})$ endowed with the topology of uniform convergence on compacts. Furthermore,  $\xi(t)$ is a centered Gaussian process such that 
\begin{equation}\label{asintotica_di_xi_diffusivo}
\lim_{t\rightarrow\infty}\frac{1}{\log(t)}E\left[\xi(t)\right]^2=
\frac{2}{\pi}\mbox{;}
\end{equation}
that is, the width of the interface fluctuations increases in time as $\log(t)$. .

However, a number of natural phenomena cannot be described by simple diffusion; e.g., the way some proteins diffuse across cell membranes or the motion of a particle in systems with geometric constraints, for example on the surface of a perfect crystal. Therefore, it can be of interest considering systems of SDEs analogous to (\ref{sistema_originale_diffusivo}) and in which the Brownian particle interacts with anomalously diffusing fields. The present paper is devoted to extending the results obtained in \cite{articolo} for system (\ref{sistema_originale_diffusivo}), to the case in which the interface fluctuations are due to  interactions with anomalously diffusing fields. In other words, we will study the long time behavior of a Brownian particle coupled with an anomalously diffusing environment 
(see systems (\ref{eqn:sistema_originale_subdiffusivo}), (\ref{eqn:sistema_originale_superdiffusivo_integrale_frazionario})
and (\ref{eqn:sistema_originale_superdiffusivo_laplac_fraz})). 

Anomalous diffusion processes are characterized by a mean square displacement which, instead of growing linearly in time, grows like
$t^{2\gamma}$, $\gamma>0,\, \gamma\neq \frac{1}{2}$.
When $0<\gamma<\frac{1}{2}$ the process is subdiffusive, when $\gamma>\frac{1}{2}$ it is superdiffusive.\\
Diffusion phenomena can be described at the microscopic level by BM and macroscopically by the heat equation, i.e. the parabolic problem associated with the Laplacian operator; the link between the two descriptions is, roughly speaking, the fact that the fundamental solution to the diffusion equation is the probability density associated with BM.\\
A similar picture can be obtained for anomalous diffusion. The main difference is that in nature a variety of anomalous diffusion phenomena can be observed and the question is how to characterize them from both the analytical and the statistical point of view. It has been shown that the microscopical (probabilistic) approach can be understood in the context of continuous time random walks (CTRW) and, in this framework, a process is uniquely determined once the probability density to move at distance $r$ in time $t$ is known (\cite{citare_per_CTRW_1}-
\cite{browniano_1}, \cite{Compte}, \cite{Klafter} and references therein). The analytical approach is based on the theory of fractional differentiation operators, where the derivative can be fractional either in time or in space
(see \cite{citare_per_differenziazione_fratta_1}-\cite{citare_per_differenziazione_fratta_3}, \cite{Taqqu} and references therein).

\nopagebreak[3]
For $f(s)$ regular enough (e.g. $f\in \C (0,t]$ with an integrable singularity at $s=0$), let us introduce the Riemann-Liouville fractional derivative,
\begin{equation}\label{operatore_differ_fratta}
D_t^{\gamma}(f):=\frac{1}{\Gamma(2\gamma)}\frac{d}{dt}\int_0^tds\frac{f(s)}{(t-s)^{1-2\gamma}},\qquad 0<\gamma<\frac{1}{2}\mbox{,}
\end{equation}
and the Riemann-Liouville fractional integral,
\begin{equation}\label{operatore_integraz_fratta}
I_t^{\gamma}(f):=\frac{1}{\Gamma(2\gamma-1)}\int_0^tds\frac{f(s)}{(t-s)^{2-2\gamma}},\qquad \frac{1}{2}<\gamma<1\mbox{,}
\end{equation}
where $\Gamma$ is the Euler Gamma function (\cite{citare_per_differenziazione_fratta_3}). Appendix \ref{app:motivation} contains a motivation for introducing such operators.
For $\frac{1}{2}<\gamma<1$ let us also introduce the fractional Laplacian $\Delta^{(\gamma)}$, defined through its Fourier transform:  if
the Laplacian corresponds, in Fourier space, to a multiplication by $-k^2$, the fractional Laplacian corresponds to a multiplication by
$-\mid k\mid^{\frac{1}{\gamma}}$. ($\ref{operatore_differ_fratta}$) and ($\ref{operatore_integraz_fratta}$) can be defined in a more general way (see \cite{citare_per_differenziazione_fratta_3}), but to our purposes the above definition is sufficient. Furthermore, notice that the operators in ($\ref{operatore_differ_fratta}$) and ($\ref{operatore_integraz_fratta}$) are fractional in time, whereas the fractional Laplacian is fractional in space.\\
Let us now consider the function $\rho^{\gamma}(t,x)$, solution to
\begin{equation}\label{eqn_soddisf_da_rho_sub_frazionaria}
\partial_t\rho^{\gamma}(t,x)=\frac{1}{\Gamma(2\gamma)}\frac{d}{dt} \int_0^t ds\,\frac{\partial^2_x\rho^{\gamma}(s,x)}{(t-s)^{1-2\gamma}} \quad\mbox{ when } 0<\gamma<\frac{1}{2}, 
\end{equation}
\begin{equation}\label{eqn_soddisf_da_rho_super_frazionaria}
\partial_t\rho^{\gamma}(t,x)=\frac{1}{\Gamma(2\gamma -1)} \int_0^t ds\,\frac{\partial^2_x\rho^{\gamma}(s,x)}{(t-s)^{2-2\gamma}} \quad\mbox{ when } \frac{1}{2}<\gamma<1 \mbox{.}
\end{equation}
It has been shown (see \cite{citare_per_CTRW_2, Compte} and references therein) that such a kernel is the asymptotic of the probability density of a CTRW run by a particle either moving at constant velocity between stopping points or instantaneously jumping between halt points, where it waits a random time before jumping again. 
On the other hand, a classic result states that the Fourier transform of the solution $\rho^{\gamma}(t,x)$ to
\begin{equation}\label{Laplaciano_frazionario}
\partial_t \rho^{\gamma}(t,x)=\frac{1}{2} \Delta^{(\gamma)}\rho^{\gamma}(t,x), \qquad \frac{1}{2} <\gamma<1, 
\end{equation}
is, for $\gamma\geq\frac{1}{2}$, the characteristic function of a (stable) process whose first moment is divergent when $\gamma\geq 1$ (see \cite{Taqqu}); this justifies the choice $\frac{1}{2}<\gamma<1$ in  equation ($\ref{Laplaciano_frazionario}$). Processes of this kind are particular CTRWs, the well known L\'evy flights; in this case large jumps are allowed with non negligible probability and this results in the process having divergent second moment.\\
We will use the notation $\rho\esp(t,x)=\rho^{\gamma}_t(x)$ to indicate the solution to 
either ($\ref{eqn_soddisf_da_rho_sub_frazionaria}$),  
($\ref{eqn_soddisf_da_rho_super_frazionaria}$) or  
($\ref{Laplaciano_frazionario}$), as in \nopagebreak[0]
the proofs we use only the properties that these kernels have in common.\\
The above described framework is analogous to the one of Einstein diffusion: for subdiffusion and Riemann-type superdiffusion the statistical description is given by 
CTRWs, whose (asymptotical) density is the fundamental solution of the evolution \hyphenation{equation}equation associated with the 
operators of fractional differentiation and integration, i.e. (\ref{eqn_soddisf_da_rho_sub_frazionaria})
and (\ref{eqn_soddisf_da_rho_super_frazionaria}), respectively (see Appendix B). 
 For the L\'evy-type superdiffusion, the statistical point of view is given by L\'evy flights, whose \-pro\-bability density evolves in time according to the evolution equation \-associated with the fractional Laplacian, i.e. (\ref{Laplaciano_frazionario}) (see \cite{Taqqu}).

In view of the above considerations, we introduce the 
following three systems of It\^o-SDEs:
\begin{equation}\label{eqn:sistema_originale_subdiffusivo}
\left\{
\begin{array}{l}
dX^{\gamma}(t)=\lambda^{\frac{1}{2\gamma}} dw(t)+ \lambda^{\frac{1}{\gamma}-1} \langle\varphi_{X^{\gamma}(t)},h^{\gamma}(t) \rangle dt\\
\\
\displaystyle
dh^{\gamma}(t)= \frac{1}{\Gamma(2\gamma)}\frac{d}{dt}\int_0^tds\frac{\partial^2_x h^{\gamma}(s,x)}
{(t-s)^{1-2\gamma}} dt- 
\varphi_{X^{\gamma}(t)} dX^{\gamma}(t),
\end{array}
\right.
\end{equation}
\\
\begin{equation}\label{eqn:sistema_originale_superdiffusivo_integrale_frazionario}
\left\{
\begin{array}{l}
dX^{\gamma}(t)=\lambda^{\frac{1}{2\gamma}} dw(t)+ \lambda^{\frac{1}{\gamma}-1} \langle\varphi_{X^{\gamma}(t)},h^{\gamma}(t) \rangle dt\\
\\
\displaystyle
dh^{\gamma}(t)=\frac{1}{\Gamma(2\gamma-1)}\int_0^tds\frac{\partial^2_x h^{\gamma}(s,x)}{(t-s)^{2-2\gamma}} dt- \varphi_{X^{\gamma}(t)} 
dX^{\gamma}(t),
\end{array}
\right.
\end{equation}
\\
and
\\
\begin{equation}\label{eqn:sistema_originale_superdiffusivo_laplac_fraz}
\left\{
\begin{array}{l}
dX^{\gamma}(t)=\lambda^{\frac{1}{2\gamma}} dw(t)+ \lambda^{\frac{1}{\gamma}-1} \langle\varphi_{X^{\gamma}(t)},h^{\gamma}(t) \rangle dt\\
\\
dh^{\gamma}(t)= \frac{1}{2}\Delta^{(\gamma)}h^{\gamma}(t) dt- \varphi_{X^{\gamma}(t)} dX^{\gamma}(t)\mbox{.}
\end{array}
\right.
\end{equation}
Roughly  speaking, the first two systems are obtained from (\ref{sistema_originale_diffusivo}), by replacing the Laplacian of the field $h(t,x)$ in equation 
$(\ref{sistema_originale_diffusivo})_2$, with the fractional derivative and fractional integral of $\Delta h(t,x)$, respectively (see (\ref{eqn_soddisf_da_rho_sub_frazionaria})
and (\ref{eqn_soddisf_da_rho_super_frazionaria})).  The last system is obtained by replacing the Laplacian with the fractional Laplacian (see (\ref{Laplaciano_frazionario})). In this way we model our anomalously diffusing fields. \\
Again, $w(t)$ is a one dimensional BM,  $\varphi(x)$ is a function in the Schwartz class and  $\varphi_{X^{\gamma}(t)}=\varphi(x-X^{\gamma}(t))$. A more detailed motivation for introducing  the above systems can be found in Appendix B.\\
We shall denote by $X\esp(t)$ the solution to either of the three above systems
(the reason for adopting  this notation, which might at first seem confusing,  will be apparent in few lines).
For $\lambda\in (0,1)$, let us introduce the scaled variables
$$
X\espp(t):=X^{\gamma}\left(t\lambda^{-\frac{1}{\gamma}}\right),
$$
$$
h\espp(t,x):=\frac{1}{\lambda} h\esp\left(x\lambda^{-1},t\lambda^{-\frac{1}{\gamma}}\right),
$$
$$
\varphi^{(\lambda)}(x):=\frac{1}{\lambda} \varphi\left(x\lambda^{-1}\right).
$$
For the function $\varphi$ only, we use the convention $\varphi_a(x):=\varphi(x-a), a\in\R$ and we set
\begin{equation}\label{def fi_t lambda}
\varphi_t^{(\lambda)}(x):=\varphi^{(\lambda)}_{\lambda X\espp}
=\frac{1}{\lambda} \varphi\left(x\lambda^{-1}- X\espp(t)\right)
\mbox{;}
\end{equation}
the notation for $\varphi^{(\lambda)}_t$ 
should include the superscript $ ^{\gamma}$, which we omit.\\
Let also  
$\xi^{\gamma}(t)$ be the solution to the integral equation
\begin{equation}\label{eqn_integrale_xi_anomalo}
\xi^{\gamma}(t)=b(t)-\int_0^t ds \rho^{\gamma}_{t-s}(0)\xi^{\gamma}(s), \qquad  \xi^{\gamma}(0)=0, \qquad0<\gamma<1,
\end{equation}
where $b(t)=\lambda^{\frac{1}{2\gamma}}w(t\lambda^{-\frac{1}{\gamma}})$. Notice that, in virtue of the scaling property of Brownian motion, the dependence of $\xi\esp(t)$ on $\lambda$ through $b(t)$ is only formal.
The main result presented in this paper is a scaling limit (in fact, three scaling limits) of $X\espp(t)$ to $\xi\esp (t)$.
Also, the solution to (\ref{eqn_integrale_xi_anomalo}) is unique by basic facts on the theory of Volterra integral equations, which we shall recall at the beginning of Section \ref{lemmata}. 
\newtheorem{teorema}{Theorem}
\begin{teorema}[first version]\label{version 1}
With the notation introduced above, we have that $\forall \gamma\in(0,1)$ and $\forall N\in[1,\infty)$ there exists $\tau=\tau(N,\gamma)>0$ such that
$$
\lim_{\lambda\rightarrow 0}E\sup_{t\leq\tau|\log\lambda|^{\frac{1}{C(\gamma)}}}
\vert X\espp(t)-\xi\esp(t)\vert^N=0,
$$    
where $C(\gamma)$ is a positive constant, with $C(1/2)=1$. 
\end{teorema} 
The fact that $C(1/2)=1$ is coherent with (\ref{risultato_diffusivo}). In Section \ref{sec: proof of thm 3} we prove an equivalent version of Theorem \ref{version 1}, namely Theorem \ref{main_result}, where the constant $C(\gamma)$ is made explicit. 
Theorem \ref{version 1} says that the asymptotic behavior of $X\espp(t)$, the rescaled solution  to either one of the systems (\ref{eqn:sistema_originale_subdiffusivo}), (\ref{eqn:sistema_originale_superdiffusivo_integrale_frazionario})  and (\ref{eqn:sistema_originale_superdiffusivo_laplac_fraz}), is described by the function $\xi\esp(t)$. Hence,  we need to determine the behavior of $\xi\esp(t)$ for large $t$, and this is the content of the following Theorem \ref{thm:theorem 2}. 
\begin{teorema}\label{thm:theorem 2}
For $\gamma\in(0,\frac{1}{2})$, 
$\xi^{\gamma}(t)$  is a centered Gaussian process s.t.  
\begin{equation}\label{asintotica_di_xi_subdiffusivo}
\lim_{t\rightarrow\infty}
E\left[\xi^{\gamma}(t)\right]^2 = const\mbox{.}
\end{equation}
For $\gamma\in\left(\frac{1}{2},1\right)$, we prove an invariance principle for $\xi^{\gamma}(t)$.  Let $\xi^{\gamma}_{\epsilon}(t)=\epsilon^{\gamma-\frac{1}{2}} \xi^{\gamma}(\epsilon^{-1}t)$; then, as $\epsilon\rightarrow0$,  $\xi^{\gamma}_{\epsilon}$ converges weakly in $\mathcal{C}(\R_+)$ to a mean zero Gaussian process, $Z(t)$, whose covariance function is
$$
E(Z(s)Z(t))=\frac{\sin^2(\pi\gamma)}{\pi^2 c(\gamma)^2}\int_0^{t\wedge s}du
\frac{1}{(t-u)^{1-\gamma}(s-u)^{1-\gamma}}\mbox{.}
$$
\end{teorema}
Intuitively, this means that in the case in which the particle interacts with a 
subdiffusive field, the feedback force exerted by the field keeps the process localized. On the other hand, the superdiffusive field (no matter which one of the two we consider) is not strong enough to contrast the effect of the Brownian nature of the particle and the width   of the fluctuation increases in time as $t^{2\gamma-1}$.\\
Notice also that the CTRW associated with the operators $D_t\esp$ and $I_t\esp$ are non-Markovian whereas  L\'evy processes are Markovian processes; nevertheless the limiting process (\ref{eqn_integrale_xi_anomalo}) is  non-Markovian for any $\gamma\in(0,1)$: in the case of L\'evy-type superdiffusion there is loss of Markovianity.\\
The paper is organized as follows: in Section \ref{sec:notation}, after establishing the notation, we state a second (equivalent) version of Theorem \ref{version 1}. This version is the one that we shall actually prove in Section \ref{sec: proof of thm 3}. Section \ref{lemmata} contains all the technical estimates used in Section \ref{sec: proof of thm 3}. This proof is a generalization of the one used in \cite{articolo} in order to prove (\ref{risultato_diffusivo}).
Section \ref{proof of theorem 2} is devoted  to the proof of Theorem \ref{thm:theorem 2}, which relies on the use of Tauberian Theorems. 
  Finally, Appendix A provides a sketch of the proof of \-exi\-stence, uniqueness and continuity of the paths of the solution to (\ref{eqn:sistema_originale_subdiffusivo}), (\ref{eqn:sistema_originale_superdiffusivo_integrale_frazionario}) and 
(\ref{eqn:sistema_originale_superdiffusivo_laplac_fraz}).
 Appendix B contains a more detailed motivation for 
the introduction of the operators of fractional differentiation and integration. \section{Notation and Results} \label{sec:notation}

The kernels in ($\ref{eqn_soddisf_da_rho_sub_frazionaria}$) and 
($\ref{eqn_soddisf_da_rho_super_frazionaria}$) can be explicitly written both in integral form (see Appendix \ref{app:motivation})
 \begin{equation}\label{espressione_esplicita_per_rho}
\rho^{\gamma}(t,x)=\frac{1}{4\pi i} \int_{c-i\infty}^{c+i\infty}dz\, e^{zt}\,
\frac{e^{-\mid x\mid z^{\gamma} }}{z^{1-\gamma}} \quad \forall c>0 \mbox{ and } 0<\gamma<1 
\end{equation}
and as a series,
\begin{eqnarray}\label{sviluppo_in_serie_della_rho}
\rho^{\gamma}(t,x)&=&\frac{1}{ 2 t^{\gamma}} M\left(\frac{\mid x\mid}{t^{\gamma}},\gamma\right),
\quad\quad0<\gamma<1,\mbox{ where} \nonumber\\
M(z,\gamma)&:=&\sum_{k=0}^{\infty} \frac{(-1)^k \, z^k}{k!\, \Gamma(-\gamma(k+1)+1)}\mbox{.}
\end{eqnarray}
The asymptotic behavior  of the \textit{Mainardi function} $M(z,\gamma)$ as $z\rightarrow +\infty$ is known,
$$
M(z,\gamma)\simeq A(\gamma)\,z^{\frac{2\gamma-1}{2-2\gamma}}\,e^{-B(\gamma)\,
z^{\frac{1}{1-\gamma}}},
$$
with $A$ and $B$ constants depending on $\gamma$ (\cite{citare_per_differenziazione_fratta_2}); hence $\rho^{\gamma}(t,x)$ has finite moments of any orders, which is
\begin{equation*}
\int_{\R} dx \,\rho^{\gamma}(t,x) \mid x\mid^n < \infty, \qquad\forall n\in N\mbox{.}
\end{equation*}
We remark that this property holds when $\rho^{\gamma}(t,x)$ is the fundamental solution of either ($\ref{eqn_soddisf_da_rho_sub_frazionaria}$) or 
($\ref{eqn_soddisf_da_rho_super_frazionaria}$).
On the other hand, the fundamental solution of ($\ref{Laplaciano_frazionario}$), namely
\begin{equation}\label{Espressione_esplicita_per_rho_laplaciano}
\rho^{\gamma}(t,x)=\int_{\R} e^{-\frac{1}{2}t\mid k\mid^{\frac{1}{\gamma}}} e^{i k x} dk\, ,\quad \gamma \in  (1/2,1),
\end{equation}
%
%
has finite first moment but divergent second moment.\\
We want to remark that in order to prove Theorem \ref{version 1} (i.e. Theorem \ref{main_result}),  we basically use only the 
following elementary properties enjoyed by both ($\ref{espressione_esplicita_per_rho}$) and ($\ref{Espressione_esplicita_per_rho_laplaciano}$):
\begin{itemize}
\item scaling property:
\begin{equation}\label{scaling_property}
\rho^{\gamma}(t,x)=\frac{1}{t^{\gamma}}\rho^{\gamma}\left(1,\frac{x}{t^{\gamma}}\right)\mbox{,}
\end{equation}
from which, setting 
\begin{equation}\label{def of c(gamma)}
c(\gamma):=\rho_1\esp(0),
\end{equation}
\begin{equation}\label{rho centrata}
\rho^{\gamma}_{t-s}(0)=\frac{\rho^{\gamma}_{1}(0)}{(t-s)^{\gamma}}
=\frac{c(\gamma)}{(t-s)^{\gamma}};
\end{equation}
\item there exists a generic constant $C=C(\gamma)>0$ such that 
\begin{equation}\label{bound 1}
 \left\vert \frac{\rho_1\esp(z)}{c(\gamma)}-1\right\vert\leq C\mid z\mid^{\beta},
\footnote{This inequality can be deduced by using (\ref{sviluppo_in_serie_della_rho}) when referring to Riemann-type anomalous diffusion, see also footnote 
\ref{foot 1}. When $\rho\esp_t$ is the kernel in (\ref{Espressione_esplicita_per_rho_laplaciano}), see footnote \ref{foot 1}.}
\quad\forall \beta\in (0,1],
\end{equation}
and 
\begin{equation}\label{bound 2}
 \left\vert \frac{\rho_1\esp(z)}{c(\gamma)}-1\right\vert\leq C.\footnote{The constant that appears in this inequality is equal to $1$ when $\rho\esp$ is either the L\'evy-type kernel or the subdiffusive kernel and it depends on $\gamma$ otherwise; see again footnote \ref{foot 1}. } 
\end{equation}
\end{itemize}
For $f,g\,\in L^2([0,t])$, $f\ast g$ denotes the \textit{Volterra convolution}, namely  
$$
(f\ast g)(t):=\int_0^t ds f(t-s) g(s).
$$ 
For $m\in \mathbb{N},\,m\geq 2,\,f^{*(m)}=f\ast f^{(m-1)}$ is the convolution of $f$ with itself $(m-1)$ times, where we define $f^{\ast(1)}(t):=f(t)$.
Set $\mathbb{K}_{\gamma}(t):=\rho_t\esp (0)$ and notice that
\begin{equation}\label{K2}
\mathbb{K}_{\gamma}^{\ast(2)}(t-s)=\!\int_0^{t-s}\!\!\!\!\!\!\!\! ds'\rho\esp_{t-s-s'}(0)\rho\esp_{s'}(0)=\int_s^t \!\!\!\!ds' \rho\esp_{t-s'}(0) \rho\esp_{s'-s}(0)=
k_{(1)}(\gamma)(t-s)^{1-2\gamma}.
\end{equation}
If we iterate $n$ times,  we end up with 
\begin{align}\label{def of K}
\mathbb{K}_{\gamma}^{\ast(n+1)}(t-s)&:=
\int_s^t ds' \rho_{t-s'}^{\gamma}(0)\mathbb{K}_{\gamma}^{\ast(n)}(s'-s)\nonumber\\
&=k_{(n)}(\gamma)(t-s)^{n-(n+1)\gamma}, \quad \qquad n\geq 1,
\end{align}
where
\begin{equation}\label{def constk}
k_{(n)}(\gamma):=c(\gamma)^{n+1}\frac{\Gamma(1-\gamma)^{n+1}}{\Gamma((n+1)(1-\gamma))} \mbox{.}
\end{equation}
To obtain  the previous equality we used the fact that the Beta function $B(z,w)$ can be expressed in terms of the Euler Gamma function in the following way:
$$
B(z,w)\stackrel{def}{=} \int_0^1 ds \,s^{z-1}(1-s)^{w-1} = \frac{\Gamma(z)\,\Gamma(w)}{\Gamma(z+w)}\qquad Re(z)\mbox{, }Re(w)>0\mbox{.}
$$
In the same way, by setting
\begin{equation}\label{def Pts}
P_{t,s}^{(\lambda,\gamma)}: = \langle \varphi_t^{(\lambda)},\rho_{t-s}^{\gamma} \varphi_s^{(\lambda)} \rangle=P_{t,s}^{\ast(1)},
\end{equation}
(on the RHS we drop the superscript $\espp$ for notational convenience) we have
$$
P_{t,s}^{\ast(2)}=\int_s^t ds' P_{t,s'}\espp P_{s',s}\espp,
$$
and for $n\geq 1$
\begin{equation}\label{iter Pts}
P_{t,s}^{\ast(n+1)}:=\int_s^t ds' P_{t,s'}\espp P_{s',s}^{\ast(n)}.
\end{equation}
We further introduce
\begin{equation*}
K_{t,s}^{(\lambda,\gamma)} := \langle \varphi_t^{(\lambda)},\int_s^tdb(s')\rho_{t-s'}^{\gamma} \varphi_{s'}^{(\lambda)} \rangle,
\end{equation*}
\begin{equation}\label{F_0}
F_0^{(\lambda,\gamma)}(t):=-\int_0^tds K_{s,0}\espp,
\end{equation}
\begin{equation*}
F_1^{(\lambda,\gamma)}(t):=\int_0^t ds P_{t,s}\espp K_{s,0}\espp,
\end{equation*}
\begin{eqnarray*}
F_2^{(\lambda,\gamma)}(t)&:=&-\int_0^t ds\int_0^s ds'
 P_{t,s}\espp P_{s,s'}\espp K_{s',0}\espp\\
 &=&-\int_0^t dsK_{s,0}\espp P_{t,s}^{\ast(2)}, 
\end{eqnarray*}
and in general
\begin{align}
F_n\espp(t)&:=(-1)^{n+1}  \int_0^t dsK_{s,0}\espp P_{t,s}^{\ast(n)},
\qquad n\geq1\label{F_n1}\\
&\,=-\int_0^t ds P\espp_{t,s}F_{n-1}\espp(s)
\qquad\qquad n\geq2\label{F_n2}.
\end{align}
Via the Duhamel principle (see Lemma \ref{lemma:Principio_di_Duhamel}), systems (\ref{eqn:sistema_originale_subdiffusivo}), (\ref{eqn:sistema_originale_superdiffusivo_integrale_frazionario})
and (\ref{eqn:sistema_originale_superdiffusivo_laplac_fraz}) can be expressed in integral form by a unique system, that is:
\begin{equation}\label{sistema_Duhamel}
\left\{
\begin{array}{l}
\displaystyle
X\espp(t)=b(t)+\int_0^t ds 
\langle \varphi_s^{(\lambda)},h\espp(s)\rangle\\
\\
\displaystyle
h\espp(t)=-\int_0^t db(s)\rho_{t-s}^{\gamma}
\varphi_s^{(\lambda)}-\int_0^t ds \langle \varphi_s^{(\lambda)},
h\espp(s)\rangle \rho^{\gamma}_{t-s}\varphi_s^{(\lambda)},
\end{array}
\right.
\end{equation}
where $\gamma\in(0,1)$; in the above system  $\rho\esp_t(x)=\rho^{\gamma}(t,x)$ is either ($\ref{espressione_esplicita_per_rho}$) for $\gamma\in(0,1)$ or ($\ref{Espressione_esplicita_per_rho_laplaciano}$) for $\gamma\in(1/2,1)$. For any $f$ in the Schwartz class, $\left(\rho^{\gamma}_{t}f\right)(x)$ is a convolution in the space variable. Namely,  
 $\rho^{\gamma}_{t-s}\varphi_s^{(\lambda)}=\left(\rho^{\gamma}_{t-s}\varphi_s^{(\lambda)}\right)(x)=
 \int_{\mathbb{R}}dy\rho^{\gamma}_{t-s}(x-y)\varphi_s^{(\lambda)}(y)$.
The initial conditions  for (\ref{sistema_Duhamel}) are
$X\espp(0)=h\espp(0)=0$. In Appendix A we prove that (\ref{sistema_Duhamel}) admits a unique solution in $\mathcal{C}(\mathbb{R}_+;\mathbb{R}\times L^2(\mathbb{R}))$. 
Notice as well that from  (\ref{sistema_Duhamel}) one has
\begin{equation}\label{integrale per h}
h\espp(t)=-\int_0^t dX\espp(s)
\rho^{\gamma}_{t-s}\varphi_s^{(\lambda)}. 
\end{equation}
Following \cite{articolo} page 10, we formally iterate once both the equation for $X\espp$ and the one for $\xi^{\gamma}$, $(\ref{sistema_Duhamel})_1$ and (\ref{eqn_integrale_xi_anomalo}), respectively. Plugging ($\ref{sistema_Duhamel}$)$_2$ into  ($\ref{sistema_Duhamel}$)$_1$  and using (\ref{integrale per h}), we get 
\begin{align}
X\espp\1(t)&= b(t)\!-\!\! \int_0^t\! \!\!\!ds K_{s,0}\espp \! +\!\! \int_0^t\!\!\!\!ds \langle\varphi_s^{(\lambda)},\!
\int_0^s\!\!\!\! ds'
\langle \varphi_{s'}^{(\lambda)},\!\int_0^{s'}\!\!\!\!\!dX\espp\1(s'') \rho_{s'-s''}\esp \varphi_{s''}^{(\lambda)}     \rangle  \rho_{s-s'}^{\gamma} \varphi_{s'}^{(\lambda)} \rangle\nonumber\\
&= b(t)+F_0\espp(t)+\!\!\!\int_0^t\!\!\!ds \!\!\int_0^s\!\!\!ds'P_{s,s'}\espp \langle 
\varphi_{s'}^{(\lambda)},\! \int_0^{s'}\!\!\!\!
dX\1\espp(s'') \rho_{s'-s''}^{\gamma} \varphi_{s''}^{(\lambda)}
\rangle,\label{X_1}
\end{align}
where the subscript $\1$ is to recall  that we are considering the first iteration of $(\ref{sistema_Duhamel})_1$.
Setting
$Y\1\espp(t) = X\1\espp(t) - b(t) - F_0\espp(t)$, $Y\1\espp(t)$ solves
\begin{eqnarray}\label{eqn:Y_1}
Y\1\espp(t)&=&\int_0^t ds F_1\espp(s)+\int_0^t ds F_2\espp(s)\nonumber\\
&+& \int_0^tds\int_0^sds'P_{s,s'}\espp\langle \varphi_{s'}^{(\lambda)},
\int_0^{s'}\!\!\!\!dY\1\espp
(s'')\rho_{s'-s''}^{\gamma} \varphi_{s''}^{(\lambda)}\rangle\mbox{;}
\end{eqnarray}
observing that $ Y\1\espp(t)$ is a.s. in  $\mathcal{C}^1(\mathbb{R})$,  we can rewrite the previous expression for $ Y\1\espp(t)$ as
\begin{equation}\label{Y_punto_una_iterata}
\dot{Y}\1\espp(t) = F_1\espp(t) + F_2\espp(t)+\int_0^tds\dot{Y}\1\espp(s)\int_s^t
ds'P_{t,s'}\espp P_{s',s}\espp,
\end{equation}
hence
\begin{equation}\label{prima_iterata_di_X}
X\1\espp(t):=\int_0^tds \dot{Y}_{(1)}\espp(s)+b(t)+F_0\espp(t)\mbox{.}
\end{equation}
On the other hand, iterating the equation for $\xi^{\gamma}$ and using (\ref{K2}), we get
\begin{equation}\label{prima_iterata_di_xi}
\xi^{\gamma}\1(t)=b(t)-\int_0^tds \rho_{t-s}^{\gamma}(0) b(s)+ k_{(1)}(\gamma)\int_0^tds (t-s)^{1-2\gamma}\xi\1^{\gamma}(s)\mbox{.}
\end{equation}
We can repeat the same procedure $n$ times; for $n\geq 2$ we then have:
\begin{equation}\label{X_n}
X\en\espp(t):= b(t)+F_0\espp(t)+\int_0^tds [F_1\espp+\dots+F_n\espp](s)
+Y\espp\en(t),
\end{equation} 
 where 
\begin{equation}\label{def of Y_n}
Y\espp\en(t):=(-1)^{n+1}\int_0^t ds \langle \varphi^{(\lambda)}_s, 
\int_0^sdX\espp\en(u)\rho\esp_{s-u}\varphi_u^{(\lambda)}\rangle
\int_s^t ds' P_{s',s}^{\ast(n)}.
\end{equation}
$Y\espp\en(t)$ solves the equation
\begin{eqnarray}\label{eqn for Y_n}
Y\espp\en(t)&\!\!=\!\!&\int_0^t ds \left[F_n\espp+\dots +F_{2n}\espp\right](s)\nonumber\\
&\!\!+\!\!&(-1)^{n+1} \int_0^t ds  \langle \varphi^{(\lambda)}_s, 
\int_0^sdY\espp\en(u)\rho\esp_{s-u}\varphi_u^{(\lambda)}\rangle
\int_s^t ds' P_{s',s}^{\ast(n)},
\end{eqnarray}
so by differentiating, using the definition of $P_{t,s}\espp$ and (\ref{iter Pts}), we get
\begin{eqnarray}\label{eqn for dotY_n}
\dot{Y}\espp\en(t)&\!\!=\!\!& \left[F_n\espp+\dots +F_{2n}\espp\right](t)\nonumber\\
&\!\!+\!\!&(-1)^{n+1} \int_0^t ds \dot{Y}\espp\en(s) P_{t,s}^{\ast(n+1)}.
\end{eqnarray}
Define $A\en\esp(t)$ as
\begin{equation}\label{espress per A_n}
A\esp\en(t):=\sum_{\nu=0}^{n}(-1)^{\nu}\left( \mathbb{K}_{\gamma}^{\ast(\nu)}\ast b\right)(t) \quad  n\geq 1,\,\,\, 0<\gamma<\frac{n}{n+1},
\end{equation}
where $\left( \mathbb{K}_{\gamma}^{\ast(0)}\ast b\right)(t)$ is only formal and we set it to be equal to $b(t)$.
Then, at the $n-$th iteration of the equation for the limiting process $\xi\esp(t)$, we find that $\forall n\geq1$,
\begin{equation}\label{xi_n}
\xi\en\esp(t)=A\en\esp(t)+(-1)^{n+1}\int_0^t ds\xi\en\esp(s)\mathbb{K}^{\ast(n+1)}(t-s).
\end{equation}
When we write $X\espp\en$, we refer to the expression (\ref{X_n}) if 
$n\geq 2$ and to (\ref{prima_iterata_di_X}) if $n=1$. As for $Y\espp\en$ and $\dot{Y}\espp\en$, expressions (\ref{eqn for Y_n}) and (\ref{eqn for dotY_n}) coincide with (\ref{eqn:Y_1}) and (\ref{Y_punto_una_iterata}) respectively, when $n=1$. So   $Y\espp\en$ and $\dot{Y}\espp\en$ are defined by (\ref{eqn for Y_n}) and (\ref{eqn for dotY_n}) $\forall n\geq 1$.\\

To prove convergence of $X\espp$ to $\xi\esp$ we prove convergence of the $n$-th iterates.  More precisely,  we prove that $\forall n\geq 1$,  $X\en\espp$ converges to $\xi\en\esp$ (in the sense of Theorem \ref{main_result} below) when $ \gamma\in\left(0,\frac{n}{n+1}\right)$. 

The reason why we consider successive iterates of the equation for $X\espp$ (and hence for $\xi\esp$) is to gain integrability and some sort of regularity. Notice indeed that $\int_0^t db(s)\rho_{t-s}\esp(0)$ is not well defined for $\gamma\geq 1/2$, whereas $\forall n\geq 1 $
\begin{equation}\label{well pos of stoch int}
\int_0^t db(s)\K_{\gamma}^{*(n+1)}(t-s)\quad\mbox{ is well defined for }\gamma\in
\left(0,\frac{2n+1}{2(n+1)}\right).\
\end{equation}
$\forall n\geq 1 $, we further restrict the range of $\gamma$ to $\gamma \in \left(0,\frac{n}{n+1}\right)$ in view of (\ref{def of K}) (see Remark \ref{cosa dimostri} and (\ref{stima Psi-n}), as well). 
\begin{teorema}[id est, second version of Theorem \ref{version 1}]\label{main_result}
With the notation introduced above, we have that 
$\forall
 \gamma\in\left(0,\frac{n}{n+1}\right)$ and 
 $\forall N\in[1,\infty)$, $\exists \tau=\tau(N,\gamma)>0$ s.t. 
\begin{equation}\label{risultato_anomalo_sub}
\lim_{\lambda\rightarrow0}E\sup_{t\leq\tau\mid\log\lambda\mid^{\frac{1}{(n+1)(1-\gamma)}}}\mid X\espp\en(t)-\xi\esp\en(t)\mid^N=0.
\end{equation}
\end{teorema}
$\|\cdot\|_p,\, p\geq 1$, indicates the usual $L^p(\mathbb{R}, dx)$ norm and $(\rho_{t}\esp f)(x)=\int dy \rho_t\esp(x-y)f(y)$ is a convolution in space. Now a few observations: $\forall t>0$ and $\forall n\geq 1$
\begin{equation}\label{varphi=varphi_n}
\varphi_t^{(\lambda)}=\frac{1}{\lambda}\varphi\left(x\lambda^{-1}-X\espp(t)\right)=
\frac{1}{\lambda}\varphi\left(x\lambda^{-1}-X\espp\en(t)\right),
\qquad \gamma\in(0,1);
\end{equation}
so actually the notation for $\varphi_t^{(\lambda)}$, as well as the one for $K_{t,s}\espp$ and $\Gamma_s\esppp$, the latter defined in (\ref{Gamma_s}), should explicitly show the ''dependence" on $n$, but we omit it. This also explains why in some estimates (for example (\ref{eqn:lemma3-3})), $n$ appears on the right hand side but not on the left hand side. 

$\forall p\geq1$ there exists a positive constant $C=C(p)$ s.t.
\begin{equation}\label{stima_di_varphi}
\|\varphi_t^{(\lambda)}\|_p\leq C\lambda^{\frac{1}{p}-1}\,\mbox{.}
\end{equation}
Moreover, $\forall t>0$, 
\begin{equation}\label{stima sul kernel}
\rho\esp(t,x)\leq B(\gamma)\rho\esp(t,0),
\end{equation}
where $B(\gamma)=1$ if $\rho\esp$ is either the 
subdiffusive kernel or ($\ref{Espressione_esplicita_per_rho_laplaciano}$), and it is a positive constant actually depending on $\gamma$ in the case of Riemann-superdiffusion.
\footnote{A more detailed account and helpful plots of the kernels (\ref{espressione_esplicita_per_rho}) can be found in \cite{citare_per_differenziazione_fratta_2} on page 1473; see also \cite{citare_per_differenziazione_fratta_3}. As for the kernel in (\ref{Espressione_esplicita_per_rho_laplaciano}), we recall that both $\rho_t\esp$ and its first derivative in space belong to $L^1(\R)\cap L^{\infty}(\R)$, $\forall t>0$ and we refer to \cite{Taqqu}.\label{foot 1}}
(\ref{stima sul kernel}) implies that
\begin{equation}\label{P_stimato_con_p_subdiff}
P_{t,s}\espp \leq B(\gamma)\rho_{t-s}\esp(0),\qquad \forall\, 0<s<t,
\end{equation}
and 
\begin{equation}\label{stima utile}
\langle \varphi\espl,\rho\esp_t\varphi\espl\rangle\leq B(\gamma)\rho\esp_t(0), \qquad\forall t>0.
\end{equation}
From (\ref{P_stimato_con_p_subdiff}), we also have
\begin{equation}\label{P < K}
P_{t,s}^{\ast(n)}\leq C\mathbb{K}_{\gamma}^{\ast(n)}(t-s),
\end{equation}
where $C>0$ is a generic constant depending on $\gamma$.

\section{ Technical Lemmata}\label{lemmata}
Throughout the following Lemma we will make extensive use of the Volterra convolution introduced in Section \ref{sec:notation}. Notice that this convolution is commutative and that it enjoys the property 
\begin{equation}\label{prop of Volt convol}
\left[\left(\int_0^{\cdot} du f(u)\right)\ast g\right](t)=\int_0^t du (f\ast g)(u),
\end{equation} 
which easily follows after a change of variable. Indeed
\begin{align*}
&\left[\left(\int_0^{\cdot} du f(u)\right)\ast g\right](t)=\int_0^t 
\!\!\!ds \left(\int_0^{t-s}
\!\!\!\!\!\!du f(u)\right) g(s)\\
=&\int_0^t \!\!\!ds g(s)\int_s^t
\!\!\!dv f(v-s)=\int_0^t\!\!\!dv\int_0^v\!\!\!ds f(v-s)g(s)\\
=&\int_0^t dv (f\ast g)(v).
\end{align*}
\begin{lemma}\label{lemma:equazioni_integrali}
For $n\in\mathbb{N},\,n\geq1,$  consider the integral equation
\begin{equation}\label{IE}
h\en(t)-(\mathbb{K}_{\gamma}^{\ast(n+1)}\ast h\en)(t)=g(t),\qquad g\in L^2([0,t]),\,
\gamma\in\left(0,\frac{n}{n+1}\right).
\end{equation}
Call $\xin\en\esp(t)$ the solution 
to (\ref{IE}) when the forcing $g(t)$ is taken to be equal to  $\A\en\esp(t)\in L^2([0,t])$ and $\etan\en\esp(t)$ the solution to the same equation with a different forcing, say $\G\esp\en (t)$. Namely:
\begin{equation}\label{csi_n}
\xin\en\esp(t)+ (-1)^n\int_0^t ds \xin \en\esp(s)k\en(\gamma)(t-s)^{n-(n+1)\gamma}=\A\esp\en(t)
\end{equation}
and 
\begin{equation}\label{eta-n-iterata}
\etan\en\esp(t)+ (-1)^n\int_0^t ds \etan\en\esp(s)k\en(\gamma)(t-s)^{n-(n+1)\gamma}=\G\esp\en(t).
\end{equation}
If the two forcings $\A\esp\en(t)$ and $\G\esp\en(t)$ are related through
\begin{equation}\label{rel tra A_n and G_n}
(\A\esp\en\ast\mathbb{K}_{\gamma}^{\ast(n+1)})(t)=(-1)^{n+1}\int_0^t ds \G\esp\en(s),
\end{equation}
then the same relation holds true between the corresponding solutions, i.e.:
\begin{equation}\label{rel tra csi_n e eta_n}
(\xin\esp\en\ast\mathbb{K}_{\gamma}^{\ast(n+1)})(t)=(-1)^{n+1}\int_0^t ds \etan\esp\en(s).
\end{equation}
\end{lemma}
The proof of this lemma is an immediate consequence of some basic facts in the theory of Volterra integral equations, which we recall here. For more details on this theory we refer the reader to \cite{Tricomi}.  For some $T>0$, let  $g(t),\mathcal{K}(t) \in L^2([0,T])$. Then the solution $h(t)$ to the equation
$$
h(t)-  \int_0^t ds \mathcal{K}(t-s) h(s)=g(s)
$$
exists and is unique and can be expressed as
\begin{equation}\label{representation_formula}
h(t)=g(t)-\int_0^t ds H(t-s) g(s),
\end{equation}
where
$$
H(t-s)=-\sum_{\nu=0}^{\infty} \mathcal{K}^{\ast(\nu + 1)}(t-s). 
$$
When the kernel $\mathcal{K}(t)$ is not in $L^2$, the solution to (\ref{representation_formula}) still exist and is unique provided that for some  $n\in\mathbb{N}$  the iterated kernel $\mathcal{K}^{\ast(n)}$ is bounded on $[0,T]$. The proof of this fact can be found in \cite{Tricomi}, Section $1\cdot12$, where kernels of the form $\mathcal{K}(t)=t^{\alpha}$, with $\alpha\in(0,1)$ are considered. 
\begin{proof}[Proof of Lemma \ref{lemma:equazioni_integrali}] 
For $\gamma \in \left(0,n/n+1\right)$, the kernel of equations (\ref{csi_n}) and  (\ref{eta-n-iterata}) is a bounded continuous function, so the standard theory for kernels in $L^2$ applies.   
Thanks to (\ref{representation_formula}), together with  (\ref{csi_n}), (\ref{eta-n-iterata}) and (\ref{rel tra A_n and G_n}), proving (\ref{rel tra csi_n e eta_n}) boils down to proving 
$$
(-1)^{n+1}\int_0^t ds \left(H\ast \G\esp\en\right)(s)=\left(\mathbb{K}_{\gamma}^{\ast(n+1)}\ast H\ast \A\esp\en\right)(t).
$$
Such an equality holds true because, by the commutativity of the Volterra convolution,  the right hand side is equal to
\begin{align*}
\left[H\ast \left(\mathbb{K}_{\gamma}^{\ast(n+1)}\ast \A\esp\en\right)\right](t)&=\int_0^t H(t-s)
\left(\mathbb{K}_{\gamma}^{\ast(n+1)}\ast \A\esp\en\right)(s)\\
&=(-1)^{n+1}\int_0^t ds H(t-s)\int_0^s \G\esp\en(s')ds';
\end{align*}
now the conclusion follows from property (\ref{prop of Volt convol}).
\end{proof}
In the following lemma and throughout the paper we will be using the notation $\mathcal{F}\{f(x)\}(k)=\hat{f}(k)$ and $\mathcal{L}\{g(t)\}(\mu)=g^{\#}(\mu)$ for the Fourier and the 
Laplace transform \-respectively and we will superscript $\tilde{}$ for the Fourier-Laplace transform.
\begin{lemma}\label{lemma:Principio_di_Duhamel}
For $0<\gamma<\frac{1}{2}$, let $v^{\gamma}(t,x)$ be a solution to
\begin{equation}\label{equaz_subdiffusiva}
\left\{
\begin{array}{ll}
\displaystyle
\partial_t v^{\gamma}(t,x)=\frac{1}{\Gamma(2\gamma)}\frac{d}{dt} \int_0^t ds\,\frac{\Delta_x v^{\gamma}(s,x)}{(t-s)^{1-2\gamma}}
+ F(t,x)
& \left(0,\infty\right)\times \R\\
\\
 v^{\gamma}(0,x)=v^{\gamma}_0(x)& \left\{0\right\}\times \R\nonumber\\
\end{array}
\right.
\end{equation}
and, for $\frac{1}{2}<\gamma<1$, let it be a solution to
\begin{equation}\label{equaz_superdiffusiva}
\left\{
\begin{array}{ll}
\displaystyle
\partial_t v^{\gamma}(t,x)=\frac{1}{\Gamma(2\gamma-1)} \int_0^t ds\,\frac{\Delta_x v^{\gamma}(s,x)}{(t-s)^{2-2\gamma}}
+ F(t,x)
&  \left(0,\infty\right)\times \R\\
\\
 v^{\gamma}(0,x)=v^{\gamma}_0(x),& \left\{0\right\}\times \R\nonumber\\
\end{array}
\right.
\end{equation}
where 
$v^{\gamma}_0(x)\in \mathcal{C}(\R)$, $F(t,x)\in\mathcal{C}(\R_+\times\R)$.  Then
\begin{equation*}
v^{\gamma}(t,x)=\int_\R dy \rho^{\gamma}(t,x-y) v^{\gamma}_0(y) +\int_0^t\! ds\int_\R dy \rho^{\gamma}(t-s,x-y)F(s,y),
\end{equation*}
with $\rho^{\gamma}(t,x)$ the kernel defined in  (\ref{espressione_esplicita_per_rho}).
\end{lemma}
\begin{proof} Let us observe that Duhamel principle for the heath equation (i.e. the parabolic equation associated with the Laplacian) can be expressed as follows: if $u(t,x)$ is a classical  
solution to            
\begin{equation}\label{equaz_calore}
\left\{
\begin{array}{ll}
\partial_t u(t,x)=\frac{1}{2}\Delta_x u(t,x)+ F(t,x)&  \left(0,\infty\right)\times \R, \qquad F\in\mathcal{C}(\R_+\times\R),\\
u(0,x)=u_0(x)& \left\{0\right\}\times \R,\;\;\;\;\,\, \qquad u_0\in \mathcal{C}(\R),\nonumber\\
\end{array}
\right.
\end{equation}
then its Fourier-Laplace transform satisfies
\begin{equation}\label{four-laplace_soluz_calore}
\tilde{u}(\mu,k)= \frac{\hat{u}(0,k)+\tilde{F}(\mu,k)}{\mu+\frac{1}{2}k^2},
\end{equation}
where $(\mu+k^2/2)^{-1}$ is the Fourier-Laplace transform of the fundamental solution of the diffusion equation, i.e. of the heat kernel.\\
Now let us recall that the Fourier-Laplace transform of $\rho^{\gamma}(t,x)$ is given by ($\ref{trasf_di_Four-Laplace_di_rho}$) (in ($\ref{trasf_di_Four-Laplace_di_rho}$) take $c_1=1$, see Appendix \ref{app:motivation}); \nopagebreak[0]
also, $ \mu^{1-2\gamma} \, \tilde{v}\esp(\mu,k)$ is the Laplace transform of $D_t^{\gamma}(\hat{v}\esp(\cdot,k))$ when  $0<\gamma<\frac{1}{2}$, whereas for $\frac{1}{2}<\gamma<1$ it is the Laplace transform of  
$I_t^{\gamma}(\hat{v}^{\gamma}(\cdot,k))$ (see Appendix \ref{app:motivation}). 
Hence
$$
\mathcal{L}(\partial_t \hat{v}\esp(t,k))=\int_0^{\infty} dt\, e^{-\mu t }\partial_t \hat{v}(t,k)
$$
$$
=-\hat{v}(0,k)+\mu \,\tilde{v}\esp(\mu,k)=- c_1 k^2 \mu^{1-2\gamma} \tilde{v}(\mu,k)+\tilde{F}(\mu,k)
$$
$$
\Rightarrow \tilde{v}(\mu,k)=\frac{\hat{v}(0,k)+\tilde{F}(\mu,k) }{\mu+c_1 k^2 \mu^{1-2\gamma} },
$$
which is precisely what we where looking for (see (\ref{trasf_di_Four-Laplace_di_rho}) and ($\ref{four-laplace_soluz_calore}$)).
\end{proof}
\begin{lemma}\label{lemma_tecnico}
 $\forall N\geq 1$ and $0<\gamma<1$, let $p$, $q$ and $r$ be real numbers greater than 1 s.t. $p^{-1}+q^{-1}=1$, $q>\max\left\{N,r\right\}$ and $r^{-1}-q^{-1}<(2\gamma)^{-1}$. 
Let $v(\cdot)$ be an $\mathcal{F}_s$-adapted process in $\mathcal{C}(\R_+,L^r(\R))$ and $\psi$ a random variable in $L^{p}(\R)$. Then there exists a positive constant 
$C=C(q,r,\gamma)$ such that 
$$
\left( E\left\vert\langle\psi , \int_{t_1}^{t_2} \!\!\!db(s) \rho^{\gamma}_{t_2-s} v(s) \rangle \right\vert^N\right)^{\frac{1}{N}} \!\!\!\leq C (t_2-t_1)^{\nu} \left(E\|\psi\|_p^{\beta} \right)^{\frac{1}{\beta}} \left( E\sup_{t_1\leq s\leq t_2} \|v(s)\|_r^q \right)^{\frac{1}{q}}\!,
$$
for any $t_1\leq t_2$, where $\beta=\frac{N q}{q-N}$ and $\nu=\frac{1}{2}-\gamma\left(\frac{1}{r}-\frac{1}{q}\right)$.
\end{lemma}
\begin{proof}\textit{(sketch)} 
The proof is identical to the proof of Lemma 3.1 in \cite{articolo}, so we will not repeat it. The additional condition $r^{-1}-q^{-1}<(2\gamma)^{-1}$ is an integrability condition and comes from the fact that 
\begin{equation}\label{integrability condition}
\int_{t_1}^{t_2} ds \|\rho_{t_2-s}^{\gamma} \|_{r'}^2<\infty\Longleftrightarrow \frac{1}{r}-\frac{1}{q}<\frac{1}{2\gamma} \mbox{,}
\end{equation}
where $r'$ is such that $\frac{1}{r'}+\frac{1}{r}=1+\frac{1}{q}$ (see page 12 in \cite{articolo}). (\ref{integrability condition}) follows from the scaling property (\ref{scaling_property}) in the following way:
\begin{align*}
&\int_{t_1}^{t_2} ds \|\rho_{t_2-s}^{\gamma} \|_{r'}^2=
\int_{t_1}^{t_2}\frac{ds}{(t_2-s)^{2\gamma}}\left( \int_{\R}dy \rho_1^{r'}(y)(t_2-s)^{\gamma}
\right)^{\frac{2}{r'}}\\
&=C\int_{t_1}^{t_2}ds (t_2-s)^{2\gamma(\frac{1}{r'}-1)}<\infty\,\Longleftrightarrow
\,2\gamma(\frac{1}{r'}-1)>-1\\
&\qquad\qquad\qquad\qquad\qquad\qquad\qquad\,\,\,
\Longleftrightarrow \,\frac{1}{r}-\frac{1}{q}<\frac{1}{2\gamma}. 
\end{align*}
\end{proof}
\begin{remark}\label{rem:integr condition}
The extra condition $\frac{1}{r}-\frac{1}{q}<\frac{1}{2\gamma}$ is automatically satisfied when $\gamma\in(0,1/2]$. It is non empty only when $\gamma\in(1/2,1)$.
\end{remark}
In the remainder of this section and in the proof of Theorem \ref{main_result} we will very often make use of the following simple observation (sometimes without saying it explicitly).
\begin{note}\label{note:scambi}
Let $(\Omega,\mu),(\Omega', \mu')$ two (finite dimensional) measure spaces, $f:\Omega\times\Omega'\rightarrow \R$ a positive function and $m$ a real number greater or equal to 1. Suppose 
\newcommand{\io}{\int_{\Omega}d\mu(x)}
\newcommand{\iop}{\int_{\Omega'}d\mu'(y)}
$$
F(y):=\int_{\Omega}d\mu(x)f(x,y)<\infty \quad\mbox{ for a.e. } y\in\Omega \quad\mbox{ and} $$
$$
\iop \left(\io  f(x,y)\right)^m<\infty.
$$
Then
\begin{align*}
&\iop \left(\io  f(x,y)\right)^m\\
 = & \iop \left[F(y)^{m-1}\io  f(x,y)\right]\\
=&\io\iop F(y)^{m-1} f(x,y) \\ 
\leq & \left(\iop F(y)^m\right)^{\frac{m-1}{m}}\io\left(\iop f(x,y)^m\right)^{\frac{1}{m}},
\end{align*}
having applied H\"older's inequality with $m/(m-1)$ and $m$.  
Looking at the first and the last line of the above equations and dividing both sides by $\left[\int_{\Omega'} \left(\int_{\Omega} f\right)^m\right]^{\frac{m-1}{m}}$ we obtain
\begin{equation}\label{scambio generico}
\iop\left(\io f(x,y)\right)^m \leq 
\left(\io\left( \iop f(x,y)^m\right)^{\frac{1}{m}}\right)^m.
\end{equation}
When $(\Omega, \mu), (\Omega', \mu')$ are just $\R$ equipped  with the Lebesgue measure, the above inequality reads
$$
\int dy\left(\int dx  f(x,y)\right)^m\leq 
\left(\int dx \left(\int dy f(x,y)^m  \right)^{\frac{1}{m}}\right)^m.
$$  
If instead $(\Omega', \mu')$ is a probability space and $(\Omega, \mu)$ is the time interval $[0,T]$ with the Lebesgue measure, inequality (\ref{scambio generico}) implies that $\forall t \in[0,T]$ and $N\geq 1$, we have
\begin{equation}\label{noteE}
E \sup_{t\in[0,T]}\left\vert \int_0^t ds f(s)\right\vert^N\leq E \left\vert\int_0^T
ds \vert f(s)\vert\right\vert^N \leq T^N \sup_{s\in[0,T]}E \vert f(s)\vert^N.
\end{equation} 
\end{note}
In the remainder of this Section,  $C$ is  a constant that does not depend on $\lambda$ or $\delta$, although it might depend on a positive power of $T$. Also, in the proofs we assume for simplicity $T\geq1$, even though all the results are true for any $T>0$, hence they are stated in such generality. Even if we assumed $T\geq 1$, this would not be restrictive in view of the fact that the result we are concerned with is a long time result, more specifically  $T\sim \mid\log\lambda\mid$ with $\lambda \rightarrow 0$.
The case $\gamma=1/2$ is not explicitly considered in Lemma \ref{lemma:stime1} and Lemma \ref{lemma:stime2}.
\begin{lemma}\label{lemma:stime1}
$\forall N\geq 1$, $ 0<\gamma< 1$ and  
$\zeta\in\left(0,\frac{1}{2\gamma}\right)$,  there exists  $C>0$ such that:
\begin{equation}\label{eqn:stima_su_K}
\sup_{0\leq s\leq t\leq T} \left(E \left\vert K_{t,s}\espp\right\vert^N \right)^{\frac{1}{N}}\leq C T^{\zeta\gamma}\lambda^{\frac{1}{2\gamma}-1-\zeta}, \qquad  T>0, \lambda\in(0,1).
\end{equation}
Also,  $\forall n\geq 1$, $N,\gamma, \zeta$ as above
\begin{equation}\label{stima_su_X_lambda_sub}
\left(E \sup_{t\in[0,T]}\left\vert X\espp\en(t) \right\vert^N\right)^{\frac{1}{N}}
\!\!\!\!\leq C \left(1\!+\lambda^{\frac{1}{2\gamma}-\zeta-1} \right)e^{CT^{(n+1)(1-\gamma)}},  
\end{equation}
$T>0, \lambda\in(0,1)$.
Moreover, for the displacement of the center we find: 
\begin{align}\label{stima_spostamento_centro_sub}
&\!\!\!\!\!\!\!\!\!\!\!\!\!\!\!\!\!\!\!\!\!\!\!\!
\left(\sup_{t\in[0,T],t+\tau\leq T}  E \sup_{t'\in[t,t+\tau]} \mid X\en\espp(t')-
X\en\espp(t)\mid^N \right)^{\frac{1}{N}}\nonumber\\
&\leq C  \left(\tau^{\frac{1}{2}} +\tau\lambda^{\frac{1}{2\gamma}-\zeta-1} \right)e^{CT^{(n+1)(1-\gamma)}},
\end{align}
$\tau,\lambda\in(0,1)$ and $T>0$.\\
\end{lemma}
\begin{proof}\textit{(sketch)} 
($\ref{eqn:stima_su_K}$) follows from Lemma \ref{lemma_tecnico} and ($\ref{stima_di_varphi}$), where in Lemma \ref{lemma_tecnico} we have chosen $\frac{1}{r}-\frac{1}{q}=\frac{1}{2\gamma}-\zeta$, $\zeta\in \left(0,\frac{1}{2\gamma}\right)$.
Having in mind Note \ref{note:scambi}, 
from (\ref{F_0}) and ($\ref{eqn:stima_su_K}$), using (\ref{noteE}) we have 
\begin{equation}\label{attesa_del_sup_di_F0}
\left(E \sup_{t\in[0,T]} \left\vert F_0^{(\lambda,\gamma)}(t)\right\vert^N \right)^{\frac{1}{N}}\leq C T^{1+\zeta\gamma} 
\lambda^{\frac{1}{2\gamma}-\zeta-1}. 
\end{equation}
From (\ref{F_n1}), (\ref{eqn:stima_su_K}) and (\ref{P < K}) we get
\begin{equation}\label{sup_dell'attesa_di_Fn}
\left(\sup_{t\in[0,T]} E\left\vert F_n\espp(t)\right\vert^N\right)^{\frac{1}{N}}
\leq T^{\zeta\gamma}T^{n-n\gamma}\lambda^{\frac{1}{2\gamma}-\zeta-1}, \quad n\geq 1 \mbox{,}
\end{equation}
so that, again by (\ref{noteE}), 
\begin{equation}\label{attesa del sup di Fn}
\left(E\sup_{t\in[0,T]}\left\vert\int_0^t ds  F_n\espp(s)\right\vert^N\right)^{\frac{1}{N}}
\leq CT^{\zeta\gamma+1}T^{n-n\gamma}\lambda^{\frac{1}{2\gamma}-\zeta-1}.
\end{equation}
Also, from (\ref{eqn for dotY_n}) and  (\ref{P < K}), 
\begin{align*}
\left(\sup_{t\in[0,T]}E\left\vert \dot{Y}\en\espp(t)\right\vert^N\right)^{\frac{1}{N}}\!\!\!\!\!
&\leq C \left\{\sup_{t\in[0,T]}\left( E\left\vert F_n\espp(t)\right\vert^N\right)^{\frac{1}{N}} \!\!\!\!\!+\dots+\!\!\sup_{t\in[0,T]}
\left( E\left\vert F_{2n}\espp(t)\right\vert^N\right)^{\frac{1}{N}}\right\}\\
&+CT^{n-\gamma(n+1)}
\int_0^Tdt\left(\sup_{s\in[0,t]}
E\left\vert \dot{Y}\en\espp(s)\right\vert^N\right)^{\frac{1}{N}}. 
\end{align*}
By the Gronwall Lemma and (\ref{sup_dell'attesa_di_Fn}) we then obtain that $\forall n\geq 1$
\begin{equation}\label{sup_dell'attesa_di_Ypunto}
\left(\sup_{t\in[0,T]}E\left\vert \dot{Y}\en\espp(t)\right\vert^N\right)^{\frac{1}{N}} \leq C \lambda^{\frac{1}{2\gamma}-\zeta-1} 
e^{CT^{(n+1)-\gamma(n+1)}},
\end{equation}
hence 
\begin{equation}\label{attesa del sup Ypunto}
\!\!\!\left(E\sup_{t\in[0,T]}\left\vert\int_0^t ds  
\dot{Y}\en\espp(s)\right\vert^N\right)^{\frac{1}{N}}\!\!\!\!\!
\leq C \lambda^{\frac{1}{2\gamma}-\zeta-1} 
e^{CT^{(n+1)-\gamma(n+1)}}.
\end{equation}
When $n=1$ ($\ref{stima_su_X_lambda_sub}$) is a straightforward consequence of ($\ref{prima_iterata_di_X}$), (\ref{attesa_del_sup_di_F0}), (\ref{attesa del sup Ypunto}) and the fact that
\begin{equation}\label{attesa sup b(t)}
E\sup_{t\in[0,T]}\mid b(t)\mid^N\leq CT^{\frac{N}{2}}.
\end{equation}
When $n>1$, we first rewrite (\ref{X_n}) as follows 
\begin{equation}\label{X_n con dotY}
X\en\espp(t)= b(t)+F_0\espp(t)+\int_0^tds \left[F_1\espp+\dots+F_n\espp\right](s)
+\int_0^t ds\dot{Y}\espp\en(s),
\end{equation}
and then 
($\ref{stima_su_X_lambda_sub}$) follows from (\ref{attesa_del_sup_di_F0}), (\ref{attesa sup b(t)}), (\ref{attesa del sup di Fn}) and (\ref{attesa del sup Ypunto}).
By acting in a similar way we find the following estimates:
$$
\left(\sup_{t\in[0,T], t+\tau\leq T} E \sup_{t'\in[t,t+\tau]} \left\vert F_0\espp(t')-F_0\espp(t)
\right\vert^N\right)^{\frac{1}{N}}\leq C \tau\, T^{\zeta\gamma}\lambda^{\frac{1}{2\gamma}-1-\zeta},
$$
$$
\left(\sup_{t\in[0,T], t+\tau\leq T} E \sup_{t'\in[t,t+\tau]}\left\vert\int_t^{t'}
ds F_n\espp(s)\right\vert^N\right)^{\frac{1}{N}}\leq C 
\tau\, T^{\zeta\gamma}T^{n(1-\gamma)}\lambda^{\frac{1}{2\gamma}-1-\zeta},
$$
$$
\left(\sup_{t\in[0,T], t+\tau\leq T} E \sup_{t'\in[t,t+\tau]}\left\vert\int_t^{t'}
ds \dot{Y}\en\espp(s)\right\vert^N\right)^{\frac{1}{N}}\!\leq C 
\tau\, T^{\zeta\gamma}T^{(2n+1)(1-\gamma)}\lambda^{\frac{1}{2\gamma}-1-\zeta}.
$$
So, recalling that for the BM $b(t)$
\begin{equation}\label{attesa sup b(t')-b(t)}
E\sup_{t'\in[t,t+\tau]}\mid b(t')-b(t)\mid^N\leq C\tau^{\frac{N}{2}},
\end{equation}
(\ref{stima_spostamento_centro_sub}) follows.\\
\end{proof}
\begin{lemma}\label{lemma:stime2}
$\forall N,n\geq 1$,  $ 0<\gamma< 1$, 
$\zeta\in\left(0,\frac{1}{2\gamma}\right)$, $T>0,\, \lambda, \delta\in(0,1)$ there exists a constant $C>0$ such that
\begin{equation}\label{eqn:lemma3-2}
\!\!\!\!\!\!\!\!\left(\!\!E \sup_{t\in[\delta,T]}  \left\vert \int_{\delta}^{t}ds K_{s,s-\delta}\espp\right\vert^N\right)
^{\frac{1}{N}}\!\!\!\!
\leq C \left[\delta^{1-\gamma}\!\!
+ \delta^{\zeta} \lambda^{\frac{1}{2\gamma}-\zeta-1}
\left(\delta^{\frac{1}{2}}+\delta  \lambda^{\frac{1}{2\gamma}-\zeta-1}\right)\right]
e^{CT^{(n+1)(1-\gamma)}},
\end{equation}
\begin{equation}\label{eqn:lemma3-1}
\!\!\!\!\!\!\!\!\left(\!\!E\! \sup_{t\in[\delta,T]}\! \left\vert\!\int_{\delta}^{t}\!\!\! ds \rho_{t-s}^{\gamma}(0)K_{s,s-\delta}\espp\right\vert^N \!\right)^{\frac{1}{N}}
\!\!\!\!\!\!\leq \!
C \!\left[\delta^{1-\gamma}\!\!
+ \delta^{\zeta} \lambda^{\frac{1}{2\gamma}-\zeta-1}\!
\left(\delta^{\frac{1}{2}}\!\!+\!\delta  \lambda^{\frac{1}{2\gamma}-\zeta-1}\right)\!\right]\!
e^{CT^{(n+1)(1-\gamma)}}\!\!,
\end{equation}
\begin{equation}\label{eqn:lemma 3-9}
\left(E \sup_{t\in[0,T]}\! \left\vert\int_0^{t} ds \rho_{t-s}^{\gamma}(0)K_{s,0}\espp\right\vert^N \right)^{\frac{1}{N}}\!\!\!\!\!\leq C\lambda^{\frac{1}{\gamma}-2\zeta-2}e^{CT^{(n+1)(1-\gamma)}}.
\end{equation}
Moreover, $\forall M>0$, we have
\begin{equation}\label{eqn:lemma3-3}
\!\!\!\!\!\sup_{s\in[\delta,T]}\!\! \left(\!\!E \left\vert \Gamma_s^{(\lambda,\delta,\gamma)}
\right\vert^N\!\right)^{\frac{1}{N}}\!\!\!\!\leq\! C\!\left[\lambda^{\frac{3}{4}}\!\!+\!
\lambda^{\frac{1}{2\gamma}-\frac{1}{4}}\!\!+\!\lambda^{\frac{M}{8}}\!\left(T^{1-2\gamma}
\textbf{1}_{_{\left\{0<\gamma<1/2\right\}}}\!\!+\!
\delta^{1-2\gamma}\textbf{1}_{_{\left\{1/2<\gamma<1\right\}}}\!\right)
\right]\!e^{CT^{(n+1)(1-\gamma)}}\!\!,
\end{equation}
where
$\textbf{1}$ is the indicator function and 
\begin{equation}\label{Gamma_s}
\Gamma_s^{(\lambda,\delta,\gamma)}:=\langle
\varphi_s^{(\lambda)} ,\int_0^{s-\delta}\!\!\!\!\!\!\!
db(s') \rho_{s-s'}^{\gamma}\varphi_{s'}^{(\lambda)} \rangle -
\int_0^{s-\delta}\!\!\!\!\!\!\!
db(s') \rho_{s-s'}^{\gamma}(0).
\end{equation} 
\end{lemma}
\begin{proof}
The proof of the bounds (\ref{eqn:lemma3-2}), (\ref{eqn:lemma3-1}), (\ref{eqn:lemma 3-9}) and (\ref{eqn:lemma3-3}) is done by following  \cite{articolo}, pages 16-18, so it shall not be very detailed.   
Recalling (\ref{varphi=varphi_n}), we have that  $\forall n\geq 1$,  $\gamma \in (0,1)$ and for $ 0\leq s\leq t,  \, K_{t,s}\espp$ can be  expressed as 
\begin{align}\label{Kprimo}
& K_{t,s}\espp
=\langle\varphi^{(\lambda)} ,\int_{s}^t\!\!
db(s')\rho_{t-s'}^{\gamma}\varphi^{(\lambda)} \rangle \\ \label{Ksecondo} 
&+\langle\varphi_{\lambda \left[X\en\espp(t)-X\en\espp(s)\right]}^{(\lambda)}-\varphi^{(\lambda)} , \int_{s}^t\!\!
db(s') \rho_{t-s'}^{\gamma}\varphi^{(\lambda)} \rangle\\ 
&+ \langle\varphi_{t}^{(\lambda)}, \int_{s}^t\!\!
db(s') \rho_{t-s'}^{\gamma}(\varphi_{s'}^{(\lambda)}- \varphi_{s}^{(\lambda)}) \rangle.\label{Kterzo}
\end{align} 
Observe also that  $\forall a,b \in \R$ and $\forall m\geq 1$
\begin{equation}\label{varphib-varphia}
\|\varphi_b-\varphi_a\|_m\leq \|\varphi'\|_m \mid b-a\mid\mbox{, }\varphi_h:=\varphi(x-h).
\end{equation}  
Let us start with proving (\ref{eqn:lemma 3-9}). We decompose $K_{s,0}\espp$ according to the prescription (\ref{Kprimo})-(\ref{Kterzo}); recalling the notation (\ref{def fi_t lambda}), the term  coming from  (\ref{Ksecondo}) becomes $$
\langle \varphi_s^{(\lambda)}-\varphi^{(\lambda)}, \int_0^s
db(s')\rho_{s-s'}\esp\varphi^{(\lambda)}\rangle. 
$$
Using Lemma \ref{lemma_tecnico}, we have 
$$
\left(E\left\vert \langle \varphi_s^{(\lambda)}-\varphi^{(\lambda)}, \int_0^s
db(s')\rho_{s-s'}\esp\varphi^{(\lambda)}
 \rangle \right\vert^N\right)^{\frac{1}{N}}\leq s^{\nu} E\|\varphi_s^{(\lambda)}-\varphi^{(\lambda)} \|_p \,\,\|\varphi^{(\lambda)} \|_r
$$
with $r,p$ and $\nu$ to be chosen according to Lemma \ref{lemma_tecnico}.
By  (\ref{varphib-varphia}), (\ref{stima_di_varphi}),  and (\ref{stima_su_X_lambda_sub}), we obtain that
$$
\!\sup_{s\in[0,t]} \!\!\left(\!\!E\left\vert \langle \varphi_s^{(\lambda)}\!-\!\varphi^{(\lambda)}
\!, \!\int_0^s\!\!\!\!
db(s')\rho\esp_{s-s'}\varphi^{(\lambda)}
 \rangle \right\vert^N\right)^{\frac{1}{N}}\!\!\!\!\!\!\leq\! C t^{\nu} \lambda^{\frac{1}{p}-1}
\lambda^{\frac{1}{r}-1}\lambda^{\frac{1}{2\gamma}-\zeta-1}
e^{Ct^{(n+1)(1-\gamma)}}
$$
$$
\qquad\qquad\qquad
\qquad\qquad\qquad\qquad\,\,\,\,\,\,
\leq Ce^{Ct^{(n+1)(1-\gamma)}}t^{\gamma \zeta}\lambda^{\frac{1}{\gamma}-2\zeta-2},
$$
having chosen $\frac{1}{r}+\frac{1}{p}-1=\frac{1}{2\gamma}-\zeta$.
For $p'$ and $q'$ such that $1/p'+1/q'=1$, we have
$$
\left\vert\int_0^{t}ds \rho_{t-s}^{\gamma}(0)\langle \varphi_s^{(\lambda)}-\varphi^{(\lambda)}, \int_0^s
db(s')\rho_{s-s'}\esp\varphi^{(\lambda)}\rangle\right\vert^N
$$
$$
\leq C \left(\int_0^t\!\!\!\frac{ds}{(t-s)^{\gamma p'}}\!\right)^{\frac{N}{p'}}\!
\left(\int_0^t\!\!\!\!ds \left\vert\langle \varphi_s^{(\lambda)}-\varphi^{(\lambda)}
\!\!, \!\int_0^s
db(s')\rho_{s-s'}\varphi^{(\lambda)}
 \rangle\right\vert^{q'}\right)^{\frac{N}{q'}}\!\!\!,
$$
so that
$$
\left( E\!\sup_{t\in[0,T]}\!\left\vert\int_0^{t} ds \rho_{t-s}^{\gamma}(0)
\langle \varphi_s^{(\lambda)}-\varphi^{(\lambda)}, \int_0^s
db(s')\rho_{s-s'}\esp\varphi^{(\lambda)}\rangle
\!\right\vert^N\!\right)^{\frac{1}{N}}
$$
$$
\leq C\! \left\vert\! E\!\left(\!\int_0^T \!\!\!\!\!ds\!
\left\vert\langle \varphi_s^{(\lambda)}\!-\!\varphi^{(\lambda)}\!,\! \int_0^s\!\!\!\!
db(s')\rho_{s-s'}\varphi^{(\lambda)}
 \rangle\right\vert^{q'}\right)^{\frac{N}{q'}}
 \right\vert^{\frac{1}{N}}
$$
$$
\leq C \left( E\left\vert  \int_0^T \!\!\!ds
\left\vert\langle \varphi_s^{(\lambda)}-\varphi^{(\lambda)}, \int_0^s
db(s')\rho_{s-s'}\varphi^{(\lambda)}
 \rangle\right\vert^{q'}  \right\vert^N\right)^{\frac{1}{Nq'}}
$$
$$
\leq C \!\!\sup_{s\in[0,T]} \!\!\left(\!\! E
\left\vert\langle \varphi_s^{(\lambda)}-\varphi^{(\lambda)}, \int_0^s
db(s')\rho_{s-s'}\varphi^{(\lambda)}
 \rangle\right\vert^{Nq'}\right)^{\frac{1}{Nq'}}\!\!\!\!\!\!\leq\! C\lambda^{\frac{1}{\gamma}
 -2\zeta-2}e^{CT^{(n+1)(1-\gamma)}}.
$$
The addends (\ref{Kprimo}) and (\ref{Kterzo}) can be examined in the same way, so we leave it to the reader. We now very briefly show how to obtain (\ref{eqn:lemma3-2}). We decompose again $K_{s,s-\delta}\espp$ according to (\ref{Kprimo})-(\ref{Kterzo}). For the term coming from   (\ref{Kprimo}),  by exchanging the order of integration (which is now allowed) and integrating by parts, we get
$$
\left(E \sup_{t\in[\delta,T]} \left\vert\int_{\delta}^{t}ds
\langle\varphi^{(\lambda)}, \int_{s-\delta}^sdb(s') \rho_{s-s'}^{\gamma}\varphi^{(\lambda)} \rangle \right\vert^N
\right)^{\frac{1}{N}}\leq C\left(\delta^{1-\gamma}+\delta^{\frac{1}{2}}\delta^{1-\gamma}\right).
$$
For the term coming from (\ref{Ksecondo}), we have
$$
\left(E \sup_{t\in[\delta,T]} \left\vert\int_{\delta}^{t}ds
\langle\varphi_{\lambda \left[X\en\espp(s)-X\en\espp(s-\delta)\right]}^{(\lambda)}-\varphi^{(\lambda)} , \int_{s-\delta}^s\!\!\!\!\!
db(s') \rho_{s-s'}^{\gamma}\varphi^{(\lambda)} \rangle\right\vert^N
\right)^{\frac{1}{N}}
$$
$$
\leq C \delta^{\zeta}\lambda^{\frac{1}{2\gamma}-\zeta-1}
\left[\delta^{\frac{1}{2}}+\delta  \lambda^{\frac{1}{2\gamma}-\zeta-1}\right]
e^{CT^{(n+1)(1-\gamma)}}, 
$$
having applied Lemma \ref{lemma_tecnico} with the choice 
$\frac{1}{r}-\frac{1}{q}=\frac{1}{2\gamma}-\zeta$, $\zeta\in\left(0,\frac{1}{2\gamma}\right)$, and (\ref{stima_spostamento_centro_sub}), as well. In an analogous way, for the term coming from (\ref{Kterzo}) we obtain
$$
\left(E \sup_{t\in[\delta,T]} \left\vert\int_{\delta}^{t}ds
\langle\varphi_{s}^{(\lambda)}, \int_{s-\delta}^s\!\!\!\!\!
db(s') \rho_{s-s'}^{\gamma}\left(\varphi_{s'}^{(\lambda)}- \varphi_{s-\delta}^{(\lambda)}\right) \rangle\right\vert^N\right)^{\frac{1}{N}}
$$ 
$$
\leq C \delta^{\zeta}\lambda^{\frac{1}{2\gamma}-\zeta-1}
\left[\delta^{\frac{1}{2}}+\delta  \lambda^{\frac{1}{2\gamma}-\zeta-1}\right]
e^{CT^{(n+1)(1-\gamma)}}.
$$
(\ref{eqn:lemma3-1}) results from applying again the same technique so we won't present the proof. 

In order to prove ($\ref{eqn:lemma3-3}$), let us express  
$\Gamma_s^{(\lambda,\delta,\gamma)}$ as
$$
\Gamma_s^{(\lambda,\delta,\gamma)}=\int dx\varphi\left(x-X\espp\en(t)\right)I\esppp_t(x),
$$   
where
$$
I\esppp_t(x):=\int_0^{t-\delta}\!\!\!\!db(s) \int_{\R} dy \varphi(y) \left\{ \rho_{t-s}\esp\left[
\lambda(x-y-X\espp\en(s))\right]- \rho_{t-s}\esp(0)\right\}.
$$
By a change of variables and using the scaling property (\ref{scaling_property}), we can rewrite
$$
I\esppp_t(x)=\!\!\int_0^{t-\delta}\!\!\!\!\!\!\!db(s)\rho\esp_{t-s}(0)\int_R \!\!dy \varphi(y) 
\left\{ \frac{1}{c(\gamma)}\rho\esp_1\left[\frac{\lambda
\left(x-y-X\espp\en(s)\right)}{(t-s)^{\gamma}}\right]
-1\right\},
$$
where $c(\gamma)$ is defined in (\ref{def of c(gamma)}).
We now use the bounds (\ref{bound 1}) and (\ref{bound 2}). More precisely, setting $z= \lambda\left(x-y-X\espp\en(s)\right)/(t-s)^{\gamma}$, we estimate the integrand above in the following way :
$$
\left\{
\begin{array}{cc}
  \left\vert\frac{\rho\esp_1(z)}{c(\gamma)}-1 \right\vert\leq C& \mbox{when} \mid x\mid>\lambda^{-1/8}\\
\\ 
 \left\vert\frac{\rho\esp_1(z)}{c(\gamma)}-1 \right\vert\leq C \mid z\mid
 &\mbox{when} \mid x\mid\leq\lambda^{-1/8}.
\end{array}
\right.
$$
So, following \cite{articolo}, pages 15-16, we apply the Burkholder 
inequality (\cite{Revuz_Yor}) and we get 
$$
E\left\vert I_t(x)\right\vert^N\leq C\,\textbf{1}_{\left\{\mid x\mid>\lambda^{-1/8}\right\}}
\left\vert \int_0^{t-\delta} \!\!\!\!\frac{ds}{(t-s)^{2\gamma}}\right\vert^{\frac{N}{2}}
\qquad\qquad\qquad\qquad\quad
$$
$$
\quad+C\,\textbf{1}_{\left\{\mid x\mid\leq\lambda^{-1/8}\right\}} E\left\vert\int_0^{t-\delta}\!\!\!\!\!\!
\frac{ds}{(t-s)^{2\gamma}}
\left(\!\!\int_{\R}dy\varphi(y)\frac{\lambda\left\vert x-y-X\espp\en(s)\right\vert}{(t-s)^{\gamma}}
\right)^2\right\vert^{\frac{N}{2}}
$$  
$$
\leq C \textbf{1}_{\left\{\mid x\mid>\lambda^{-1/8}\right\}}
\left\vert \textbf{1}_{ _{\left\{0<\gamma<1/2\right\}}}t^{1-2\gamma}+
\textbf{1}_{ _{\left\{1/2<\gamma<1\right\}}}\delta^{1-2\gamma}
\right\vert^{\frac{N}{2}}
$$
$$
+ C \textbf{1}_{\left\{\mid x\mid\leq\lambda^{-1/8}\right\}}
\lambda^N \left(\lambda^{-N/8}+1+E\sup_{s\in[0,T]}
\left\vert X\espp\en\right\vert^N \right)
\left\vert \int_0^{t-\delta}\frac{ds}{(t-s)^{4\gamma}}\right\vert^{\frac{N}{2}}
$$
$$
\leq C \textbf{1}_{\left\{\mid x\mid>\lambda^{-1/8}\right\}}
\left\vert \textbf{1}_{ _{\left\{0<\gamma<1/2\right\}}}t^{1-2\gamma}+
\textbf{1}_{ _{\left\{1/2<\gamma<1\right\}}}\delta^{1-2\gamma}
\right\vert^{\frac{N}{2}}
$$
$$
+C\textbf{1}_{\left\{\mid x\mid\leq\lambda^{-1/8}\right\}}\lambda^N
\left( \lambda^{-N/8}+1 +\lambda^{\frac{1}{2\gamma}-\zeta-1}\right)
e^{CT^{(n+1)(1-\gamma)}},
$$
where in the last inequality we have used (\ref{stima_su_X_lambda_sub}). If we choose $\zeta=1/8$ in the above, we obtain
$$
\left(E\left\vert \int_{ _{\left\{\mid x\mid\leq\lambda^{-1/8}\right\}}}
\!\!\!\!\!\!\!\!\!\!\!\!\!\!\varphi\left(x-X\espp\en(t)\right) I_t\esppp
\right\vert^N \right)^{\frac{1}{N}}
$$
$$
\leq C\|\varphi\|_{\frac{N}{N-1}}
\left(E\int_{ _{\left\{\mid x\mid\leq\lambda^{-1/8}\right\}}}
\!\!\!\!\!\!\!\!\!\!\!\!\!\!\!\!\!\!\!\!
dx \left\vert I_t\esppp\right\vert^N
\right)^{\frac{1}{N}}
$$
$$
\leq C \left(\lambda^{3/4}+\lambda^{\frac{1}{2\gamma}-\frac{1}{4}}\right)
e^{CT^{(n+1)(1-\gamma)}}.
$$
Moreover, for any $M>0$ we have
$$
\left(E\left\vert\int_{ _{\left\{\mid x\mid>\lambda^{-1/8}\right\}}}
\!\!\!\!\!\!\!\!\!\!\!\!\!\!\!\varphi\left(x-X\espp\en(t)\right) I_t\esppp
\right\vert^N\right)^{\frac{1}{N}}
$$
$$
\leq\lambda^{\frac{M}{8}}\left[E\left(\int_{ _{\left\{\mid x\mid>\lambda^{-1/8}\right\}}}
\!\!\!\!\!\!\!\!\!\!\!\!\!\!\!dx\varphi\left(x-X\espp\en(t)\right)^{\frac{2N}{2N-1}} (1+x^2)^{\frac{1}{2N-1}} \mid x\mid^{\frac{2NM}{2N-1}}
\right)^{2N-1}\right]^{\frac{1}{2N}}
$$
$$
\cdot \left(E\int_{ _{\left\{\mid x\mid>\lambda^{-1/8}\right\}}}
\!\!\!\!\!\!\!\!\!\!\!\!\!\!\!dx\frac{\left\vert I_t\esppp(x)\right\vert^{2N}}{1+x^2}
\right)^{\frac{1}{2N}}
$$
$$
\leq \lambda^{M/8}\left(  \textbf{1}_{ _{\left\{0<\gamma<1/2\right\}}}t^{1-2\gamma}+
\textbf{1}_{ _{\left\{1/2<\gamma<1\right\}}}\delta^{1-2\gamma}   \right).
$$
This concludes the proof of (\ref{eqn:lemma3-3}).
\end{proof}
\begin{lemma}\label{lemma:P-rho}
$\forall \,0\leq s\leq t$, $\lambda\in(0,1), \,\beta \in(0,1]$, 
$n,N\geq 1$ and $\gamma\in(0,1)$ we have
\begin{equation}\label{eqn:Lemma 7-1}
\left(E\left(\sup_{0\leq s\leq t\leq T}(t-s)^{(1+\beta)\gamma }
\vert P_{t,s}\espp-\rho_{t-s}\esp(0)\vert\right)^N\right)^{\frac{1}{N}}\leq C
\lambda^{\beta}\lambda^{\frac{1}{2\gamma}-\zeta-1}
e^{CT^{(n+1)(1-\gamma)}}.
\end{equation}
Also,  $\forall \delta\in(0,1)$ and for any $Q>0$, we have 
\begin{equation}\label{eqn:Lemma 7-2}
\left(E\sup_{t\in[\delta,T]}\left(\int_0^{t-\delta}\!\!\!\!ds\vert P_{t,s}\espp-\rho_{t-s}\esp(0)\vert\right)^N
\right)^{\frac{1}{N}}\!\!\leq C \left(\lambda^{\frac{1}{2\gamma}-\zeta}+\lambda^Q\right)
e^{CT^{(n+1)(1-\gamma)}}.
\end{equation}
\end{lemma}
\begin{proof}[Sketch of proof]
Using the definition of $P_{t,s}\espp$(\ref{def Pts}), by change of variables and the scaling property 
(\ref{scaling_property}), we have 
\begin{align}\label{eqn:P-rho scalato}
&\vert P_{t,s}\espp-\rho_{t-s}\esp(0)\vert\nonumber\\
\leq&\rho_{t-s}\esp(0)\!\!\int\!\!\!\!\int\!\!\!dx dy\,\varphi(x)
\varphi(y)
\left\vert \frac{1}{c(\gamma)}\rho\esp_1\!\!
\left[\!\frac{\lambda(x-y+\!X^{\espp}_{\en}\!(t)\!-\!X^{\espp}_{\en}\!(s))}{(t-s)^{\gamma}}\right]
\!\!-\!1\right\vert.
\end{align}
From (\ref{bound 1}), then
$$
\vert P_{t,s}\espp-\rho_{t-s}\esp(0)\vert\leq C\rho_{t-s}\esp(0)\!\!
\int\!\!\!\!
\int \!\!dxdy 
\varphi(x)\varphi(y)\frac{\left\vert x-y+X\espp\en(t)-X
\espp\en(s)\right\vert^{\beta}\lambda^{\beta}}
{(t-s)^{\gamma\beta}}.
$$
We now want to use (\ref{stima_spostamento_centro_sub}) in order to conclude; though, (\ref{stima_spostamento_centro_sub}) holds only for $N\geq1$ whereas $\beta$ is in the range  $\beta\in (0,1]$. We don't want to choose $\beta=1$ (see (\ref{prob integrabilita}) and comments after it), hence we first apply Young inequality with $p=1/\beta$ and get
$$
\vert P_{t,s}\espp-\rho_{t-s}\esp(0)\vert\leq C\rho_{t-s}\esp(0)
\int\!\!\!\!
\int \!dxdy 
\varphi(x)\varphi(y)\frac{\left(\lv x\rv+\lv y\rv +\lv X\espp\en(t)-X
\espp\en(s)\rv+1\right)\lambda^{\beta}}
{(t-s)^{\gamma\beta}},
$$ 
and now  (\ref{eqn:Lemma 7-1}) is a straightforward 
consequence of (\ref{stima_spostamento_centro_sub}). To get (\ref{eqn:Lemma 7-2}), we use again the bounds (\ref{bound 1}) and (\ref{bound 2}), this time in the following way:
setting $z= \lambda\left(x-y+X\espp\en(t)-X\espp\en(s)\right)/(t-s)^{\gamma}$, we estimate $$
\left\{
\begin{array}{cc}
  \left\vert\frac{\rho\esp_1(z)}{c(\gamma)}-1 \right\vert\leq C& \mbox{when} \mid x\mid>\lambda^{-1}\\
\\ 
 \left\vert\frac{\rho\esp_1(z)}{c(\gamma)}-1 \right\vert\leq C \mid z\mid
 &\mbox{when} \mid x\mid\leq\lambda^{-1}.
\end{array}
\right.
$$
So, from (\ref{eqn:P-rho scalato}) we have
$$
\int_0^{t-\delta}\!\!\!\!ds\vert P_{t,s}\espp-\rho_{t-s}\esp(0)\vert
$$
$$
\leq C \int_0^{t-\delta}\!\!\!\frac{ds}{(t-s)^{2\gamma}}\int\int\varphi(x)\varphi(y)
\textbf{1}_{_{\left\{\mid x\mid\leq\lambda^{-1}\right\}}}
\left[\lambda\left(x+y+X\espp\en(t)-X\espp\en(s)\right)\right]
$$
$$
+ C \int_0^{t-\delta}\!\!\!\frac{ds}{(t-s)^{\gamma}}\int\int\varphi(x)\varphi(y)
\textbf{1}_{_{\left\{\mid x\mid>\lambda^{-1}\right\}}}\!C
$$
$$
\leq C\textbf{1}_{_{\left\{\mid x\mid\leq\lambda^{-1}\right\}}} \lambda\int_0^{t-\delta}\!\!\!\frac{ds}{(t-s)^{2\gamma}}
\left(C+\left\vert X\espp\en(t)-X\espp\en(s)\right\vert\right)
$$
$$
+C\textbf{1}_{_{\left\{\mid x\mid>\lambda^{-1}\right\}}}
\int_0^{t-\delta}\!\!\!\frac{ds}{(t-s)^{\gamma}}\left(\int\varphi(x)\vert x\vert^{2Q}\right)^{\frac{1}{2}}
\left(\int_{_{\left\{\mid x\mid>\lambda^{-1}\right\}}}\frac{\varphi(x)}{\vert x\vert^{2Q}}\right)^{\frac{1}{2}}.
$$
(\ref{eqn:Lemma 7-2}) now follows  from (\ref{stima_su_X_lambda_sub}).
\end{proof}

\section{Proof of Theorem \ref{main_result}} \label{sec: proof of thm 3}
We recall that $C$ is a positive constant that does not depend on $\lambda$ and $\delta$, though it might depend on a positive power of $T$. Also, for simplicity, all the proofs are presented for $T\geq1$, even though the statements are clearly still valid for any $T>0$. Since it has already been treated in \cite{articolo}, the case $\gamma=1/2$ is not explicitly considered.

The intuitive idea that motivates the structure of the proof is based on the observation that, ``morally", things go as if $P_{t,s}\espp$ were converging to $\rho_{t-s}\esp(0)$ as  $\lambda\rightarrow 0$ (see Lemma \ref{lemma:P-rho}); formally, this can be obtained by thinking that, as  $\lambda\rightarrow 0$, $\varphi_t^{(\lambda)}\rightarrow \delta_0$. While such an idea is not hard to turn into a rigorous argument, one of the main technical difficulties is encountered when trying to do the same thing to get an intuition on what $K_{s,0}\espp$ ought to converge to. If in the definition of $K_{s,0}\espp$ we replace  $\varphi_t^{(\lambda)}$ with  $\delta_0$ and exchange the order of integration, we find that $K_{s,0}\espp$ should converge to $\int_0^t db(s)\rho_{t-s}\esp(0)$. The problem is that we are not allowed to exchange the order of integration (see comment after (3.5) in \cite{articolo}) and that $\int_0^t db(s)\rho_{t-s}\esp(0)$ is not well defined as a process in $\mathcal{C}(\R_+)$ when $\gamma\geq\frac{1}{2}$. 
In the same way, $\forall n\geq1$, $F_n\espp$ is well defined  for any $ \gamma \in (0,1)$, whereas the object it converges to is not (see (\ref{stima Fn-suo lim>1/2}) and (\ref{well pos of stoch int})).

The proof goes as follows. $\forall n\geq1$ we introduce the process 
$\eta\esp\en(t)$, solution to the equation
\begin{equation}\label{eta_n}
\eta\en\esp(t)=G\en\esp(t)+(-1)^{n+1}\int_0^t ds\eta\en\esp(s)\mathbb{K}^{\ast(n+1)}(t-s),
\quad0<\gamma<\frac{n}{n+1}
\end{equation}
where
\begin{equation}\label{G_n}
G\esp\en(t):=\sum_{\nu=n}^{2n}(-1)^{\nu +1}\!\!\int_0^t\!db(u) \mathbb{K}_{\gamma}^{\ast(\nu+1)}(t-u),
\quad  n\geq 1,\, 0<\gamma<\frac{n}{n+1}.
\end{equation}
We now observe that Lemma \ref{lemma:equazioni_integrali}
can be applied to $\xi\esp\en$, defined in (\ref{xi_n}), and $\eta\esp\en$. In this case the forcing terms are $A\en\esp$ and $G\en\esp$, respectively, and we can easily prove that they are related through (\ref{rel tra A_n and G_n}). We can in fact show that  the $i-th$ addend of $A\en\esp$ is related to the $i-th$ addend of $G\en\esp$ through (\ref{rel tra A_n and G_n}); all we need to show is that $\forall \nu\in{0,...,n,}$
$$
(-1)^{\nu}\left(\mathbb{K}_{\gamma}^{\ast(\nu)}\ast b\ast \mathbb{K}_{\gamma}^{\ast(n+1)}\right)(t)=
(-1)^{n+1}\!\!\int_0^t\!\! ds (-1)^{\nu+n+1}\!\!\int_0^s\!\!db(u) \mathbb{K}_{\gamma}^{\ast(\nu+1)}(s-u),
$$
which is a straightforward consequence of the definition of $\mathbb{K}_{\gamma}^{\ast(m)}$ given in (\ref{def of K}), together with the following equality
\begin{equation}\label{rel tra K and Theta}
\left(\mathbb{K}_{\gamma}^{\ast(n+1)}\ast b\right)(t)=\int_0^t ds\int_0^s db(u) \mathbb{K}_{\gamma}^{\ast(n)}(s-u),\qquad n\geq 1\mbox{.}
\end{equation}
Hence, Lemma \ref{lemma:equazioni_integrali} gives
\begin{equation}\label{conseq of lemma1}
(-1)^{n+1}\left(\xi\en\esp*\K_{\gamma}^{\ast(n+1)}\right)(t)=\int_0^t ds \eta\en\esp(s).
\end{equation}
Recall that the definition of $X\en\espp$ is given by (\ref{X_n}) for $n\geq 2$ and by (\ref{prima_iterata_di_X}) when $n=1$.
Using (\ref{conseq of lemma1}), we look at the difference between $X\en\espp$ and $\xi\en\esp$:
\begin{subequations}\label{X_n-xi_n}
\begin{eqnarray}
X\en\espp(t)&\!\!\!-\!\!\!&\xi\en\esp(t)=F_0\espp+\int_0^t ds \rho_{t-s}\esp(0)b(s)\\
&+&\sum_{j=1}^{n-1}\left[\int_0^t ds F_j\espp(s)-(-1)^{j+1}\left(\mathbb{K}^{\ast(j+1)}\ast b\right)(t) \right]\label{somma1}\\
&+&\left[\int_0^t\!\! ds \dot{Y}\en\espp(s)-(-1)^{n+1}\!\!\int_0^t
\!\!\!\! ds\xi\en\esp(s)\mathbb{K}^{\ast(n+1)}(t-s) \right]\\
&=& F_0\espp+\int_0^t ds \rho_{t-s}\esp(0)b(s)\label{somma d}\\
&+&\sum_{j=1}^{n-1}\left[\int_0^t ds F_j\espp(s)-(-1)^{j+1}\left(\mathbb{K}^{\ast(j+1)}\ast b\right)(t) \right]\label{somma2}\\
&+&\int_0^t ds 
\left(\dot{Y}\en\espp(s)-\eta\en\esp(s)\right)\label{somma f},
\end{eqnarray}
\end{subequations}
where for $n=1$ the sum in (\ref{somma1}) (and in (\ref{somma2})) is understood to be equal 
to zero. As we have already said, we want to prove that $\forall n\geq 1$,  $X\espp\en$ converges to $\xi\esp\en$ for $\gamma \in\left(0, \frac{n}{n+1}\right)$. To this end, let us further expand the integrand in (\ref{somma f}), using the fact that $\dot{Y}\en\espp$ solves equation (\ref{eqn for dotY_n}):
\begin{equation}\label{eqn:gronwall_su_Ypunto-eta}
\left(\dot{Y}\en\espp-\eta\esp\en\right)(t) =  R\en\espp(t)+(-1)^{(n+1)}\int_0^t
ds\left(\dot{Y}_{\lambda}\esp-\eta\esp\right)(s)\K^{\ast(n+1)}(t-s)
\end{equation}
where 
\begin{eqnarray}\label{eqn:R_n lambda}
R\en\espp(t)&:=&\sum_{j=n}^{2n}F_j\espp(t)-G\en\esp(t)\nonumber\\
&+& (-1)^{(n+1)}\int_0^tds\,\dot{Y}\en\espp(s)\left[P_{t,s}^{\ast(n+1)}
-\K^{\ast(n+1)}(t-s) \right]\mbox{,}\label{eqn:R_n lambda}
\end{eqnarray}
and $G\en\esp(t)$ is defined in (\ref{G_n}).\\
Let $\delta\in(0,1)$.
From now on we assume $t\geq \delta$.
\begin{remark}\label{t<delta}
We omit to study  the case $t<\delta$ because it can be treated in the same way as it is dealt with in \cite{articolo}, where it is presented explicitly, see in particular (3.23), (3.44) and (3.45) in  \cite{articolo}.
In other words, what we actually show is that the estimates in (\ref{Esup small A}), (\ref{stima F2 primo add})-(\ref{stima F0-suo lim}) and (\ref{stima R-n}) are valid when the supremum is taken over the interval $[\delta, T]$ (more precisely, in the case of (\ref{stima F2 primo add})-(\ref{stima F0-suo lim}) and (\ref{stima R-n}) the supremum should be over $[\lambda^a, T]$, because at that point $\delta$ will have been chosen to be equal to $\lambda^a$, see lines before (\ref{stima F0-suo lim}) ). Though, by acting as in \cite{articolo}, we can show that the same estimate holds true when the supremum is taken over the whole interval $[0, T]$. Hence from now on we will assume  $t\geq \delta$ in order to streamline the notation and the presentation of the proof.
\end{remark}
We use the decomposition
\begin{align*}
& \int_0^t ds \rho_{t-s}\esp(0)K_{s,0}\espp=\int_0^{\delta}ds \rho_{t-s}\esp(0)K_{s,0}\espp\\
 +&\int_\delta^t ds \rho_{t-s}\esp (0)\langle \varphi_s\espl, 
\int_0^{s-\delta}\!\!\!\!\!\!db(s') \rho_{s-s'}\esp\varphi_{s'}\espl\rangle
+\int_\delta^t ds \rho_{t-s}\esp (0)\langle \varphi_s\espl,
\int_{s-\delta}^s\!\!\!\!\!\!db(s') \rho_{s-s'}\esp\varphi_{s'}\espl\rangle,
\end{align*}
which follows from the definition of $K_{s,0}\espp$, to rewrite the difference between $F_1\espp$ and $\int_0^t db(s)\K^{\ast(2)}_{\gamma}(t-s)$. For 
$\gamma \in (0,1/2)$, 
\begin{align}
&\left\vert F_1\espp(t) - \int_0^t\!\!db(s) \K_{\gamma}^{*(2)}(t-s)\right\vert^N
\!\!\!\label{F1-suo lim}\\
=&\left\vert\int_0^tdsP_{t,s}\espp K_{s,0}\espp-\int_0^t\!\!db(s) \K_{\gamma}^{*(2)}(t-s)\right\vert^N
\nonumber\\
\leq& C\left\vert\int_0^tds\left(P_{t,s}\espp-\rho_{t-s}^{\gamma}(0)\right) \, K_{s,0}\espp \right\vert ^N\\
+&C\left\vert\int_0^{\delta}ds \rho_{t-s}^{\gamma}(0)\,K_{s,0}\espp \right\vert^N \nonumber\\
+&C\lv \int_{\delta}^t ds\rho_{t-s}^{\gamma}(0)\langle
\varphi_s^{(\lambda)} ,\int_0^{s-\delta}\!\!\!\!\!\!\!\!
db(s') \rho_{s-s'}^{\gamma}\varphi_{s'}^{(\lambda)} \rangle -
\int_{\delta}^t ds\rho_{t-s}^{\gamma}(0)\int_0^{s-\delta}\!\!\!\!\!\!\!\!\!
db(s')\rho_{s-s'}^{\gamma}(0)\rv^N\nonumber\\
+&C\lv\int_{\delta}^t ds\rho_{t-s}^{\gamma}(0)
\int_0^{s-\delta}\!\!\!\!\!db(s') 
\rho_{s-s'}^{\gamma}(0)-\int_0^t db(s) \K^{*(2)}(t-s)\rv^N\label{difference}\\
+&C\lv\int_{\delta}^t ds \rho_{t-s}^{\gamma}(0)\langle \varphi_s^{(\lambda)} ,\int_{s-\delta}^sdb(s') \rho_{s-s'}^{\gamma}\varphi_{s'}^{(\lambda)} \rangle
\rv^N\nonumber\\
\leq& C\left\vert\int_0^tds\left(P_{t,s}\espp-\rho_{t-s}^{\gamma}(0)\right) \, K_{s,0}\espp \right\vert ^N \label{F_1 primo addendo}\\
+& C\lv\int_0^{\delta}\!\!\!ds\rho_{t-s}^{\gamma}(0)\,K_{s,0}\espp\rv^N +C\lv \int_{\delta}^t\!ds \rho_{t-s}^{\gamma}(0) 
\Gamma_s^{(\lambda,\delta,\gamma)}\rv^N\label{F1 pezzi piccoli1}\\
+&C\lv\int_{\delta}^{t}\!
ds \rho_{t-s}^{\gamma}(0)K_{s,s-\delta}\espp\rv^N+C
\lv\Psi\1^{(\delta,\gamma)}(t)\rv^N\label{F1 pezzi piccoli2}, 
\end{align}
where in the last inequality we  used the definition of $\Gamma_s\esppp$ given in  (\ref{Gamma_s}) and we set  $\Psi\1^{(\delta,\gamma)}(t)$ to be the difference in (\ref{difference}), namely
$$
\Psi\1^{(\delta,\gamma)}(t):=\int_{\delta}^t ds\rho_{t-s}^{\gamma}(0)
\int_0^{s-\delta}\!\!\!\!\!db(s') 
\rho_{s-s'}^{\gamma}(0)-\int_0^t db(s) \K^{*(2)}(t-s), \quad\gamma\in (0,1/2).
$$
For $n\geq 1$, we define
\begin{equation}\label{Psi_n}
\Psi_{(n+1)}^{(\delta,\gamma)}(t):=\int_0^t ds\rho_{t-s}\esp(0)
\Psi_{(n)}^{(\delta,\gamma)}(s),\quad \gamma\in\left(0,\frac{n}{n+1}\right).
\end{equation}
In the same way, by using (\ref{F_n2}), (\ref{P_stimato_con_p_subdiff}) and (\ref{Psi_n}), we have
\begin{align}
&\left\vert F_2\espp(t)+ \int_0^t db(s)\K^{*(3)}(t-s)\right\vert^N
\label{F2-suo lim}\\
\leq C&\left\vert \int_0^tds\int_0^s ds'\rho _{s-s'}^{\gamma}(0) \left( P_{t,s}\espp-\rho _{t-s}^{\gamma}(0)\right)\,
K_{s',0}\espp\rv^N\label{F2 primo addendo}\\
+ C\!&\lv\int_0^t \!\!\!ds\rho_{t-s}\esp(0)\!\!\int_0^{\delta}\!\!\!\!\!ds'\rho_{s-s'}^{\gamma}(0)
K_{s',0}\espp\rv^N \!\!\!\!\!+C\!\lv\int_0^t \!\!\!ds\rho_{t-s}\esp(0)\!\! \int_{\delta}^s\!
\!\!\!ds' \rho_{s-s'}^{\gamma}(0) \Gamma_{s'}^{(\lambda,\delta,\gamma)}
\rv^N\label{F2 pezzi piccoli1}\\
+C&\lv\int_0^t ds\rho_{t-s}\esp(0)\int_{\delta}^{s}\!
ds' \rho_{s-s'}^{\gamma}(0)K_{s',s'-\delta}\espp\rv^N+C
\lv\Psi_{(2)}^{(\delta,\gamma)}(t)\rv^N\label{F2 pezzi piccoli2}. 
\end{align}
\begin{remark}\label{cosa dimostri}
We will show that the terms in (\ref{F1 pezzi piccoli1}) and the first addend in (\ref{F1 pezzi piccoli2}) (hence also the addends in
(\ref{F2 pezzi piccoli1}) and the first addend in (\ref{F2 pezzi piccoli2})) are small for $\gamma\in(0,1)$ (see (\ref{Esup small A}), (\ref{Esup small B}) and (\ref{eqn:lemma3-1})). The reason why we need to iterate the equation for $X\espp$ and $\xi\esp$ an infinite number of times comes from $\Psi^{(\delta,\gamma)}\en$ (see (\ref{well pos of stoch int}) and (\ref{stima Psi-n})). We will in fact prove that
\begin{equation}\label{stima Psi-n}
\left(E\sup_{t\in[\delta,T]}\lv \Psi\en^{(\delta,\gamma)}(t)\rv^N\right)^{\frac{1}{N}}\leq C\delta^{n-(n+1)\gamma}.
\end{equation}
Also, we will show that (\ref{F_1 primo addendo}) is small 
when $\gamma\in(0,1/2)$ and  (\ref{F2 primo addendo}) is small for 
$\gamma\in(1/2,1)$, see (\ref{stima F1primo addendo}) and 
(\ref{stima F2 primo add}) .
\end{remark}
Let us now address the points mentioned in  Remark \ref{cosa dimostri}, in the same order in which we listed them.\\
For $p,q>1$ s.t. $p^{-1}+q^{-1}=1$ and $p\gamma<1$, we have
$$
\lv\int_0^{\delta}\!\!\!ds\rho_{t-s}^{\gamma}(0)\,K_{s,0}\espp\rv^N\leq C
\lv\int_0^{\delta}\frac{ds}{(t-s)^{p\gamma}}\rv^{\frac{N}{p}}
\lv\int_0^{\delta}ds \mid K_{s,0}\espp\mid^q\rv^{\frac{N}{q}}.
$$
Since $t\geq\delta$, 
$$
\int_0^{\delta}\frac{ds}{(t-s)^{p\gamma}}\leq
 \int_0^{\delta}\frac{ds}{(\delta-s)^{p\gamma}}=C \delta^{1-p\gamma},
$$

hence
$$
E\sup_{t\in[\delta,T]}\lv\int_0^{\delta}\!\!\!ds\rho_{t-s}^{\gamma}(0)\,K_{s,0}\espp\rv^N\leq C \delta^{\frac{1-p\gamma}{p}N}E\lv \int_0^{\delta}ds\mid K_{s,0}\espp\mid^q\rv^{\frac{N}{q}}
$$
$$
\leq C \delta^{\frac{1-p\gamma}{p}N}
\left(E\lv \int_0^{\delta}ds\mid K_{s,0}\espp\mid^q\rv^N\right)^{\frac{1}{q}}\!\!\!\!
\leq C \delta^{\frac{1-p\gamma}{p}N}\delta^{\frac{N}{q}}\sup_{s\in[0,T]}
\left(E\mid K_{s,0}\espp\mid^{Nq}\right)^{\frac{1}{q}},
$$
where in the last inequality we used Note \ref{note:scambi}.
If we choose $p=\frac{\gamma+1}{2\gamma}$ and $q=\frac{\gamma+1}{1-\gamma}$, by using (\ref{eqn:stima_su_K}) we get
\begin{equation}\label{Esup small A}
E\sup_{t\in[0,T]}\lv\int_0^{\delta}\!\!\!ds\rho_{t-s}^{\gamma}(0)\,K_{s,0}\espp\rv^N\leq C \delta^{1-\gamma}\lambda^{\frac{1}{2\gamma}-\zeta-1}, \quad\gamma\in(0,1).
\end{equation}

By the same sort of trick used to get (\ref{Esup small A}), we also get
$$
\lv\int_{\delta}^sds'\rho_{s-s'}\esp(0)\Gamma\esppp_{s'}\rv^N\leq C\sup_{s'\in[\delta,s]}
\left( E \lv\Gamma\esppp_{s'}\rv^{Nq}\right)^{\frac{1}{q}}.
$$
Therefore, using (\ref{eqn:lemma3-3}), we have
\begin{align}\label{Esup small B}
&E\sup_{s\in[\delta,T]}\lv\int_{\delta}^sds'\rho_{s-s'}\esp(0)\Gamma\esppp_{s'}\rv^N
\!\!
\leq  C\left[\lambda^{\frac{3}{4}}+
\lambda^{\frac{1}{2\gamma}-\frac{1}{4}}\right]e^{CT^{(n+1)(1-\gamma)}}
\qquad\qquad\quad\qquad\qquad\nonumber\\
+&\left[\lambda^{\frac{M}{8}}(T^{1-2\gamma}
\textbf{1}_{_{\left\{0<\gamma<1/2\right\}}}+
\delta^{1-2\gamma}\textbf{1}_{_{\left\{1/2<\gamma<1\right\}}})
\right]e^{CT^{(n+1)(1-\gamma)}},\quad\gamma\in(0,1).
\end{align} 
Notice that on the right hand side of the above equation, $n$ appears because $X\espp\en$ is contained in the definition of $\Gamma\esppp_s$, see (\ref{eqn:lemma3-3}), (\ref{varphi=varphi_n}) and the comment after it.\\
As for the first term in (\ref{F1 pezzi piccoli2}) (respectively, the first term in (\ref{F2 pezzi piccoli2})), we just use (\ref{eqn:lemma3-1}) in Lemma \ref{lemma:stime2}. In order to prove (\ref{stima Psi-n}), we show in some detail how the estimate for $\Psi_{(1)}\espp$ is obtained; the way one gets (\ref{stima Psi-n}) for $n\geq 1$ should then be obvious from the definition (\ref{Psi_n}) and using (\ref{def of K}). Recalling that we are assuming $t\geq\delta$, using (\ref{K2}) and exchanging the order of integration in the definition of 
$\Psi_{(1)}\espp$  we have
\begin{align}\label{pezzi Psi_2}
\Psi_{(1)}(t)\espp&= \int_0^{t-\delta} \!\!\!\!\!\!db(s)\int_{s+\delta}^tds'
\rho\esp_{t-s'}(0)\rho\esp_{s'-s}(0)-\int_0^t\!db(s)\int_{s}^t
\!\!\!\!ds'
\rho\esp_{t-s'}(0)\rho\esp_{s'-s}(0)\nonumber\\
&=-\int_0^{t-\delta} \!\!\!\!\!\!db(s)\int^{s+\delta}_s\!\!\!\!ds'
\rho\esp_{t-s'}(0)\rho\esp_{s'-s}(0)-\int_{t-\delta}^t 
\!\!\!\!\!db(s)\int_s^t ds'\rho\esp_{t-s'}(0)\rho\esp_{s'-s}(0).
\end{align}
Now we can estimate the two terms in (\ref{pezzi Psi_2})  separately.
In both cases we first make a further change of variables and then integrate by parts the stochastic integral. We show how to handle the first, for the second the procedure is the same:
$$
\lv \int_0^{t-\delta} \!\!\!\!\!\!db(s)\int^{s+\delta}_s\!\!\!\!ds'
\rho\esp_{t-s'}(0)\rho\esp_{s'-s}(0) \rv=
\lv \int_0^{t-\delta} \!\!\!\!\!\!db(s)\int_0^{\delta}du \rho\esp_{t-s-u}(0)
\rho_u\esp(0)\rv
$$
$$
\leq\lv b(t-\delta)\int_0^{\delta}du\rho_{\delta-u}\esp(0)\rho_u\esp(0)\rv
+\sup_{s\in[0,t-\delta]}\mid b(s)\mid
\lv \int_0^{t-\delta} \!\!\!\!\!ds\frac{\partial}{\partial s}
\int_0^{\delta}du \rho\esp_{t-s-u}(0)
\rho_u\esp(0)\rv
$$
$$
\leq\lv b(t-\delta)\int_0^{\delta}du\rho_{\delta-u}\esp(0)\rho_u\esp(0)\rv
+\sup_{s\in[0,t-\delta]}\mid b(s)\mid\lv\int_0^{\delta} du
\rho_{\delta-u}\esp(0)\rho_u\esp(0)-\int_0^{\delta}du\rho_{t-u}\esp(0)\rho_u\esp(0) \rv.
$$
Notice now that from (\ref{K2}), 
$$
\int_0^{\delta}du\rho_{\delta-u}\esp(0)\rho_u\esp(0)= C \delta^{1-2\gamma}
$$
and, since $t\geq \delta$,
$$
\int_0^{\delta}\rho_{t-u}\esp(0)\rho_u\esp(0)=
 C\int_0^{\delta}\frac{du}{(t-u)^{\gamma}u^{\gamma}}\leq C\int_0^{\delta}\frac{du}
 {(\delta-u)^{\gamma}u^{\gamma}}=\int_0^{\delta}du\rho_{\delta-u}\esp(0)\rho_u\esp(0).
$$
So, after dealing with the second term in (\ref{pezzi Psi_2})  in an analogous way, (\ref{stima Psi-n}) follows by using  (\ref{attesa sup b(t)}).

Let us now turn to (\ref{F_1 primo addendo}) and (\ref{F2 primo addendo}). Let $\beta>0$; then for (\ref{F_1 primo addendo}), applying the H\"older inequality, we have
$$
\left\vert\int_0^tds\left(P_{t,s}\espp-\rho_{t-s}^{\gamma}(0)\right) \, K_{s,0}\espp \right\vert ^N
$$
\begin{equation}\label{prob integrabilita}
\leq\!\! \sup_{0\leq s\leq t} \left\{\lv P_{t,s}\espp\!-\!\rho_{t-s}
\esp(0)\rv^N\!\!\! 
(t-s)^{\gamma(1+\beta)N}\!\right\}\! \lv\int_0^t \!\!\!ds \lv K_{s,0}\espp\rv^p \rv^{\frac{N}{p}}\!\!
\lv\int_0^t \!\!\!\frac{ds}{(t-s)^{\gamma q(1+\beta)}}\rv^{\frac{N}{q}}\!\!\!\!.
\end{equation}
Looking at the last integral in (\ref{prob integrabilita}), we need to impose the integrability condition $\beta<-1+1/\gamma$. Taking the supremum for $t\in [0,T]$,  the expectation of both sides,  using (\ref{eqn:stima_su_K}) and (\ref{eqn:Lemma 7-1}), we then obtain that for $\gamma\in(0,1/2)$ and for any $N\geq 1$,
\begin{equation}\label{stima F1primo addendo}
E\sup_{t \in[0,T]}\left\vert\int_0^tds\left(P_{t,s}\espp-\rho_{t-s}^{\gamma}(0)\right) \, K_{s,0}\espp \right\vert ^N\!\!
\leq C \lambda^{\frac{1}{\gamma}-2\zeta-\frac{3}{2}}e^{CT^{(n+1)(1-\gamma)}},
\end{equation} 
where we have chosen $\beta=1/2$ in (\ref{eqn:Lemma 7-1}). We can make such a choice for $\beta$ because when we study the difference in (\ref{F1-suo lim}), and hence (\ref{F_1 primo addendo}),  we take $\gamma\in(0,1/2)$, see Remark \ref{cosa dimostri}. When we consider (\ref{F2 primo addendo}), we can't mimic what we have done for (\ref{F_1 primo addendo}); in fact from (\ref{prob integrabilita}) we get that the left hand side of (\ref{stima F1primo addendo}) is bounded by 
$\lambda^{\beta+\frac{1}{\gamma}-\zeta-2}\exp(CT^{(n+1)(1-\gamma)})$.
When we impose the integrability condition $\beta<-1+1/\gamma$ and $\beta+\frac{1}{\gamma}-\zeta-2>0$, $\beta\in(0,1]$, we find that these two conditions together cannot be satisfied for all $\gamma\in(0,1)$ (actually they hold at most for $\gamma\in(0,2/3)$). So, when $\gamma\in(1/2,1)$ we need to do something else. 
\begin{equation}\label{bla1}
\left\vert \int_0^tds\int_0^s ds'\rho _{s-s'}^{\gamma}(0) \left( P_{t,s}\espp-\rho _{t-s}^{\gamma}(0)\right)\,
K_{s',0}\espp\rv^N
\end{equation}
$$
\leq C\left\vert \int_0^{t-\delta}ds\int_0^s ds'\rho _{s-s'}^{\gamma}(0) \lv P_{t,s}\espp-\rho _{t-s}^{\gamma}(0)\rv\,
 K_{s',0}\espp\rv^N
$$
$$
+ C\left\vert \int_{t-\delta}^tds\int_0^s ds'\rho _{s-s'}^{\gamma}(0) \lv P_{t,s}\espp-\rho _{t-s}^{\gamma}(0)\rv\,
 K_{s',0}\espp\rv^N
$$
$$
\leq C \sup_{s\in[0,T]}\lv\int_0^s \!\!\!ds'\rho _{s-s'}^{\gamma}(0)K_{s',0}\espp\rv^N
$$
\begin{equation}\label{bla2}
\cdot\left(   \sup_{t\in[\delta,T]}\left\vert \int_0^{t-\delta}\!\!\!\!\!ds\lv P_{t,s}\espp-\rho _{t-s}^{\gamma}(0)\rv+
\int_{t-\delta}^t\rho_{t-s}\esp(0)\rv^N
\right),
\end{equation}
where in the last inequality we have used (\ref{eqn:P-rho scalato}) and then (\ref{bound 2}).
By (\ref{eqn:lemma 3-9}) and (\ref{eqn:Lemma 7-2}), we then have
\begin{equation}\label{stima F2 primo add}
\!\!\!\!\! \left(\!E\!\!\sup_{t\in[0,T]}\left\vert \int_0^t\!\!\!\!ds\!\!\int_0^s \!\!\!\!ds'\rho_{_{s-s'}}^{^{\gamma}}(0)\! \lv P_{_{t,s}}\espp\!\!-\!\rho_{_{t-s}}^{^{\gamma}}(0)\!\rv\,
\!\!K_{_{s',0}}\espp\!\rv^N\right)^{\frac{1}{N}}\!\!\!\!\!\!
\leq C \lambda^{^{\frac{1}{2\gamma}-\zeta-1}}\!\left(\lambda^{^{\frac{1}{2\gamma}-\zeta}}\!\!\!
+\delta^{^{1-\gamma}}\right)
e^{^{CT^{(n+1)(1-\gamma)}}}\!\!\!.
\end{equation}
If in (\ref{stima Psi-n}), (\ref{Esup small A}), (\ref{Esup small B}) and (\ref{eqn:lemma3-1}) we choose  $\delta=\lambda$ and $M>0$, recalling (\ref{stima F1primo addendo}) we have that for $\gamma\in(0,1/2)$ and $\forall N\geq 1$,  $\exists b(\gamma)>0$ s.t. 
\begin{equation}\label{stima F1-suo lim}
\left(E\sup_{t\in[0,T]}\left\vert F_1\espp(t) - \int_0^t\!\!db(s) 
\K_{\gamma}^{*(2)}(t-s)\right\vert^N\right)^{\frac{1}{N}}\leq C \lambda^{b(\gamma)}
e^{CT^{2(1-\gamma)}}.
\end{equation}
Via (\ref{F_n2}) and (\ref{P_stimato_con_p_subdiff}), this implies that for $n\geq 1$, $\gamma\in(0,1/2)$ and $\forall N\geq 1$,  $\exists\, b(\gamma)>0$ s.t. 
\begin{equation}\label{stima Fn-suo lim <1/2}
\left(E\sup_{t\in[0,T]}\left\vert F_n\espp(t) - (-1)^{(n+1)}\int_0^t
\!\!db(s) 
\K_{\gamma}^{*(n+1)}(t-s)
\right\vert^N\right)^{\frac{1}{N}}\!\!\!\leq C \lambda^{b(\gamma)}
e^{CT^{2(1-\gamma)}}.
\end{equation}
On the other hand, if in 
(\ref{stima Psi-n}), (\ref{Esup small A}), (\ref{Esup small B}) and (\ref{eqn:lemma3-1}) we chose   $\delta=\lambda^a$,
 with $a=\frac{2\gamma-1}{2\gamma(1-\gamma)}$,
and  $M>\frac{4(2\gamma-1)^2}{\gamma(1-\gamma)}$, 
recalling (\ref{stima F2 primo add}), we find that 
  $\forall n\geq 2$, $1/2<\gamma<\frac{n}{n+1}$ and $N\geq 1$, $\exists\, l(\gamma)>0$ s.t. 
\begin{equation}\label{stima Fn-suo lim>1/2}
\!\left(\!E\sup_{t\in[0,T]}\!
\left\vert F_n\espp(t)\!- \!(\!-1)^{(n+1)} \!\!\!\int_0^t\!\!\! db(s)\K^{*(n+1)}(t-s)\right\vert^N
\right)^{\frac{1}{N}}\!\!\!\!\leq C \lambda^{l(\gamma)}
e^{CT^{(n+1)(1-\gamma)}}.
\end{equation}
\begin{note}\label{note}
We want to stress that the above estimate (\ref{stima Fn-suo lim>1/2}) is needed only for $n\geq 2$ and $1/2<\gamma<\frac{n}{n+1}$, whereas (\ref{stima Fn-suo lim <1/2}) is valid for any $n\geq 1$ and  $\gamma\in(0,1/2)$. In other words we will not need an estimate on 
$\left\vert F_1\espp(t) - \int_0^t\!\!db(s) \K_{\gamma}^{*(2)}(t-s)\right\vert$ for $\gamma>1/2$. 
\end{note} 
Set now 
$$
\Psi_{(0)}(t)^{(\delta,\gamma)}:=\int_0^t ds b(s)\rho_{t-s}\esp(0)-\int_{\delta}^t ds \int_0^{s-\delta}db(s')
\rho_{s-s'}\esp(0), 
$$
then
\begin{eqnarray*}
& &\left\vert F_0\espp+\int_0^t ds b(s)\rho _{t-s}^{\gamma}(0) \right\vert^N\nonumber\\
&\leq&\!\!\!\!C\left\vert \int_0^{\delta} ds K_{s,0}^{(\lambda)}\right\vert^N
\!\!\! + C\left\vert \int_{\delta}^t ds \Gamma_s^{(\lambda,\delta,\gamma)}\right\vert^N
\!\!\!+C\left\vert\int_{\delta}^t ds K_{s,s-\delta}\espp \right\vert^N
 \!\!\!+C\lv\Psi_{(0)}(t)^{(\delta,\gamma)}\rv^N\!.\nonumber
\end{eqnarray*}
It is easy to prove that 
$$
\left(E\sup_{t\in[\delta,T]}\left\vert\Psi_{(0)}(t)^{(\delta,\gamma)}\rv^N\right)^{\frac{1}{N}}
\leq C\delta^{1/2}.
$$
So by (\ref{eqn:stima_su_K}), (\ref{eqn:lemma3-3}) and (\ref{eqn:lemma3-2}), by choosing again $\delta=\lambda^a,$
 $a=\!\textbf{1}_{_{\left\{0<\gamma<1/2\right\}}}+
 \frac{2\gamma-1}{2\gamma(1-\gamma)}\textbf{1}_{_{\left\{1/2<\gamma<1\right\}}}$
 and $M> 0\cdot\textbf{1}_{_{\left\{0<\gamma<1/2\right\}}}+
 \frac{4(2\gamma-1)^2}{\gamma(1-\gamma)} 
 \textbf{1}_{_{\left\{1/2<\gamma<1\right\}}}$ ,we get that $\forall n\geq1$, $0<\gamma<\frac{n}{n+1}$ 
and $\forall N\geq 1$, $\exists\, m(\gamma)>0$ s.t. 
\begin{equation}\label{stima F0-suo lim}
\!\!\!\left(\!E\sup_{t\in[0,T]}
\left\vert F_0\espp(t)-\int_0^t ds b(s)\rho_{t-s}\esp(0)\rv^N\right)^{\frac{1}{N}}
\leq C \lambda^{m(\gamma)}e^{CT^{(n+1)(1-\gamma)}}.\quad
\end{equation}
We will also need the following estimate:
\begin{equation}\label{stima int Psi_n}
\left(E\sup_{t\in[\delta,T]}\lv\int_0^t ds\Psi\en^{(\delta,\gamma)}(s)\rv^N\right)^{\frac{1}{N}}\leq C\delta^{(n+1)(1-\gamma)}.
\end{equation}
This inequality can be worked out with calculations analogous to those needed to obtain (\ref{stima Psi-n}), hence we omit them; roughly speaking, looking at (\ref{stima Psi-n}), (\ref{stima int Psi_n}) is correct thanks to the further integration. Also, it is what one would expect in view of the fact that $\int_0^t ds b(s) \K^{*(n+1)}(t-s)$ is defined for any 
$\gamma \in (0,1)$, as opposed to $\int_0^t db(s) \K^{*(n+1)}(t-s)$. With this remark in mind, it is easily seen that, with  the same steps that lead to an estimate on  $\left\vert F_n\espp(t)- (-1)^{(n+1)} \int_0^t db(s)\K^{*(n+1)}(t-s)\right\vert$, using this time (\ref{stima F2 primo add}) and (\ref{stima int Psi_n}), we  have that $\forall n\geq1$, $\gamma\in\left(0,\frac{n}{n+1}\right)$  and $\forall N\geq 1$, $\exists \tau=\tau(\gamma,N)>0$ s.t.
\begin{equation}\label{stima int Fn-suo lim}
\lim_{\lambda\rightarrow 0}E\!\!\sup_{t\leq \tau\lv \log\lambda\rv^{\frac{1}{(n+1)(1-\gamma)}}}
\left\vert \int_0^t\!\!F_n\espp(s)- (-1)^{(n+1)}\!\! \int_0^t \!ds\,b(s)\K^{*(n+1)}(t-s)\right\vert^N=0.
\end{equation} 
The last ingredient that we will need in order to conclude is the following estimate: $\forall n\geq1$, $\gamma\in(0,\frac{n}{n+1})$ and $\forall N\geq 1$, $\exists \,d(\gamma)>0$ s.t.
\begin{equation}\label{Ydot p-rho}
\!\!\left(E\sup_{t\in[0,T]}\lv\int_0^t\!\!\!ds\,\dot{Y}\en\espp(s)\left[P_{t,s}^{\ast(n+1)}
-\K^{\ast(n+1)}(t-s)\right]\rv^N\right)^{\frac{1}{N}}\!\!\!\!\leq C \lambda^{d(\gamma)}
e^{CT^{(n+1)(1-\gamma)}},
\end{equation}
which is  obtained by combining (\ref{eqn:Lemma 7-1}) and (\ref{attesa del sup Ypunto}) when $n=1$; when $n\geq 2$, we act like in (\ref{bla1})-(\ref{bla2})  and  then use (\ref{eqn:Lemma 7-2}) and (\ref{attesa del sup Ypunto}).\\
From the definition of  $R\en\espp$ given in  (\ref{eqn:R_n lambda}), using  (\ref{stima Fn-suo lim <1/2}), (\ref{stima Fn-suo lim>1/2}) and (\ref{Ydot p-rho}), it is straightforward to see that $\exists \,\tilde{d}(\gamma)>0$ s.t.
\begin{equation}\label{stima R-n}
\left(E\sup_{t\in[0,T]}\lv R\en\espp(t)\rv^N\right)^{\frac{1}{N}}\leq C 
\lambda^{\tilde{d}(\gamma)}
e^{CT^{(n+1)(1-\gamma)}},
\end{equation}
for any $n\geq1$, $\gamma\in\left(0, \frac{n}{n+1}\right)$ and $N\geq1$. Hence,  the  Gronwall Lemma applied to  (\ref{eqn:gronwall_su_Ypunto-eta}), gives  that $\forall n\geq 1$, $\gamma\in\left(0, \frac{n}{n+1}\right)$ and $N\geq1$, $\exists \tau=\tau(\gamma,N)>0$ s.t.
\begin{equation}\label{stima Ypunto-eta}
\lim_{\lambda\rightarrow0}   E\sup_{t\leq\tau \mid\ln\lambda\mid^{\frac{1}{(n+1)(1-\gamma)}}}
\lv \left(\dot{Y}_{\lambda}^{\gamma}-\eta^{\gamma}\right)(t)\rv^N =0\mbox{.}
\end{equation}
Finally, looking at (\ref{somma d}), (\ref{somma2}), (\ref{somma f}), thanks to  (\ref{stima F0-suo lim}), (\ref{stima int Fn-suo lim}) and (\ref{stima Ypunto-eta}), Theorem \ref{main_result} is proven.


\section{Proof of Theorem \ref{thm:theorem 2}}\label{proof of theorem 2}
In the diffusive case, the integral equation ($\ref{eqn_integrale_xi_diffusivo}$) is explicitly solvable. To our knowledge, ($\ref{eqn_integrale_xi_anomalo}$) cannot be solved for $\gamma\neq\frac{1}{2}$. However, considering the associated Green function, that is, the solution of 
\begin{equation}\label{eqn_integrale_per_greenfunction_anomala}
F^{\gamma}(t)=1-\int_0^t ds \rho_{t-s}^{\gamma}(0)F^{\gamma}(s),\quad 0<\gamma<1,
\end{equation}
one gets
\begin{equation}\label{legame_tra_xi_e_F}
\xi^{\gamma}(t)=\int_0^t db(s)F^{\gamma}(t-s),\quad 0<\gamma<1 \mbox{.}
\end{equation}
Notice that the theory of Volterra integral equations for kernels with bounded iterates implies that the solution to 
(\ref{eqn_integrale_per_greenfunction_anomala}) 
is unique, as commented at the beginning of Section \ref{lemmata}, after the statement of Lemma \ref{lemma:equazioni_integrali}. 
\begin{lemma}\label{thm:asintotica_per_F}
For any $0<\gamma<1$, the following holds:
\begin{equation}\label{asintotica_di_F}
\lim_{t\rightarrow\infty}t^{1-\gamma}F^{\gamma}(t)=\frac{\sin(\pi\gamma)}{\pi c(\gamma)}\mbox{,}
\end{equation}
where $c(\gamma)$ is defined in (\ref{def of c(gamma)}).
\end{lemma}
\begin{remark}
Since $c(1/2)=(2\pi)^{-1/2}$,  Lemma \ref{thm:asintotica_per_F}  is an extension of Theorem 2.2 in \cite{articolo}. When $\gamma=1/2$, it provides an alternative proof of such a theorem.
\end{remark}
\begin{proof}[Proof of Lemma \ref{thm:asintotica_per_F}]
By taking the Laplace transform of ($\ref{eqn_integrale_per_greenfunction_anomala}$) we obtain that  
the Green \hyphenation{function} 
function $F^{\gamma}$ has Laplace transform
\begin{equation}\label{trasf_di_Laplace_di_F}
(F^{\gamma})^{\#}(\mu)=\frac{\mu^{-\gamma}}{\mu^{1-\gamma}+c(\gamma)\Gamma(1-\gamma)}\mbox{.}
\end{equation}
Provided that $F^{\gamma}(t)$ is monotone decreasing, the Tauberian Theorem for densities (see e.g. \cite{Feller}) gives
$$
\lim_{t\rightarrow\infty}t^{1-\gamma}F^{\gamma}(t)=\frac{1}{\Gamma(\gamma)}
\lim_{\mu\rightarrow 0}\mu^{\gamma}(F^{\gamma})^{\#}(\mu)\mbox{.}
$$
Therefore the only thing we need to show is that $F^{\gamma}(t)$ is monotone decreasing. We recall that a function is completely monotone if and only if its even \hyphenation{derivatives}
derivatives are positive and the odd ones are negative. Furthermore, a function is the Laplace transform of a positive measure if and only if it is completely monotone (see again \cite{Feller}).
We think of $dF^{\gamma}(t)$ as a (a priori signed) measure on $\mathbb{R}_+$ and introduce
$$
\Phi^{\#}(\mu):=-\int_0^{\infty}e^{-\mu t}dF^{\gamma}(t)=1-\mu(F^{\gamma})^{\#}(\mu)\mbox{.}
$$
By ($\ref{trasf_di_Laplace_di_F}$) we have 
$$
\Phi^{\#}(\mu)=\frac{c(\gamma)\Gamma(1-\gamma)}{\mu^{1-\gamma}+c(\gamma)\Gamma(1-\gamma)}\mbox{.}
$$
The function $(0,\infty)\ni\mu\longrightarrow\mu^{1-\gamma}$ is positive and has completely monotone derivatives. For $A>0$ the function 
$(0,\infty)\ni x\longrightarrow A (A+x)^{-1}$ is completely monotone. Hence (see \cite{Feller}), the function $\Phi^{\#}(\mu)$ is completely monotone and we are done.
\end{proof}
\begin{proof} [Proof of Theorem \ref{thm:theorem 2}]
By (\ref{legame_tra_xi_e_F}) we get
$$
E\left[\xi\esp(t)\right]^2=\int_0^t (F^{\gamma}(s))^2 ds \mbox{,}
$$
so (\ref{asintotica_di_xi_subdiffusivo}) is straightforward.
In order to prove the invariance principle in Theorem 1, we first need to prove tightness of the process $\xi^{\gamma}_{\epsilon}(t)$. From ($\ref{legame_tra_xi_e_F}$) and ($\ref{asintotica_di_F}$) few computations show that for each $\gamma\in \left(\frac{1}{2},1\right)$ there exists a constant $C=C(\gamma)$ such that
$$
\lim_{\epsilon\rightarrow 0}E(\xi^{\gamma}_{\epsilon}(t)-\xi^{\gamma}_{\epsilon}(s))^2\leq C(t-s)^{2\gamma-1}\mbox{.}
$$
Since $\xi^{\gamma}_{\epsilon}$ is a Gaussian process, we can first obtain a bound on the higher \hyphenation{moments}
moments, thus getting tightness from the Kolmogorov's criterion.\hyphenation{convergence}
Finally, the convergence of the finite dimensional distributions follows from the convergence of the covariance, deduced from ($\ref{legame_tra_xi_e_F}$) and ($\ref{asintotica_di_F}$).
\end{proof}


\subsection*{Acknowledgments}
I am very grateful to  P. Butt\`a, who was my supervisor when I was working on the content of this paper. I thank the thoughtful referees for many comments that helped improving the paper. I am also grateful to La Sapienza University, that funded part of this work.

\appendix
\section{Existence and uniqueness}\label{app:exist uniq}
In this section we sketch the proof of existence, uniqueness and continuity of the solution of the system (\ref{sistema_Duhamel}).
\begin{teorema}
Let $B$ be the Banach space of vectors $(X,h)\in \mathbb{R}\times L^2(\R)$ with the norm
$$
\parallel(X,h)\parallel_B:=\sqrt{\mid X\mid ^2 + \parallel h\parallel_2^2}\,\mbox{.}
$$
Let us consider the following Cauchy problem with initial datum 
$(X_0,h_0)\in B$
\begin{align}\label{sistema_unicita}
\left\{
\begin{array}{l}
\displaystyle
X(t)=X_0 +b(t)+\int_0^tds\,\Upsilon(X(s),h(s))\\
\\
\displaystyle
h(t)=\rho_t^{\gamma} h_0 - 
\int_0^t\, db(s) \rho_{t-s}^{\gamma}\varphi_{X(s)}-\int_0^t ds\, 
\Upsilon(X(s),h(s))\rho_{t-s}^{\gamma}\varphi_{X(s)}\mbox{,}
\\
\end{array}
\right.
\end{align}
where $\Upsilon\mbox{:}B\rightarrow \R$ is bounded and globally Lipschitz; recall that $\varphi$ is a probability density in the Schwartz class of test functions and $\varphi_{X}=\varphi(x-X).$\\
Then for any $(X_0,h_0)\in B$  there exists a unique solution to (\ref{sistema_unicita}); such a solution, $(X(t),h(t))$, belongs to $\mathcal{C}(\mathbb{R}_+;B)$ and is such that
\begin{equation}\label{A.4}
\sup_{t\in[0,T]}E\|(X(t),h(t))\|_B^2<\infty\qquad\forall T>0\mbox{.}
\end{equation}
Uniqueness holds in the following sense: if $(\bar{X}(t),\bar{h}(t))$ is another continuous solution satisfying (\ref{A.4}), then 
\begin{equation*}
P\left(\sup_{t\in[0,T]}\|(X(t),h(t))-(\bar{X}(t),\bar{h}(t)) \|_B^2=0\right)=1
\qquad\forall T>0\mbox{.}
\end{equation*}
\end{teorema}
\begin{proof}
We prove existence by Picard iterations, uniqueness by  the Gronwall Lemma and continuity by using Kolmogorov's criterion. For the time being  $\rho_t^{\gamma}$  is either (\ref{espressione_esplicita_per_rho}) or (\ref{Espressione_esplicita_per_rho_laplaciano}), so $\gamma\in(0,1)$. \\
\textit{Existence}: construct the sequence 
$\{(X^{(n)}(t),h^{(n)}(t))\}$ such that $(X^{(0)}_t,h^{(0)}_t)=(X_0,\rho_t^{\gamma} h_0)$ and, for $n\geq1$,
\begin{equation*} 
\left\{
\begin{array}{l}
\displaystyle
X^{(n)}(t)=X_0 +b(t)+\int_0^tds\,\Upsilon(X^{(n-1)}(s),h^{(n-1)}(s))\\
\\
\displaystyle
h^{(n)}(t)=\rho_t h_0 - \int_0^t \!\!db(s) \rho_{t-s}^{\gamma}\varphi_{X^{(n-1)}(s)}-\!
\int_0^t \!\! ds\rho^{\gamma}_{t-s} \beta(X^{(n-1)}(s),h^{(n-1)}(s)),
\end{array}
\right.
\end{equation*}
where we set $\beta(X,h):=\Upsilon(X,h)\varphi_X$; notice that for a suitable constant $K>1$ we have
\begin{equation*} 
\mid\Upsilon(X,h)\mid^2+\|\beta(X,h)\|_2^2+\|\varphi_X\|_2^2\leq K
\end{equation*}
\begin{equation*} 
\lv\Upsilon(X,h)-\Upsilon(Y,g)\rv+\|\beta(X,h)-\beta(Y,g)\|_2
+\|\varphi_X-\varphi_Y\|_2\leq K\|(X,h)-(Y,g)\|_B\mbox{,}
 \end{equation*}
for any $(X,h)$ and $(Y,g)$ in $B$. Hence
$$
E\|(X^{(1)}(t), h^{(1)}(t))-( X^{(0)}(t), h^{(0)}(t) ) \|_B^2\leq 2 K^2\left(t+t^2\right)\mbox{;}
$$
moreover, by the Cauchy-Schwarz inequality,
\begin{equation*}
E\left\vert X^{^{(n+1)}}\!(t)\!-\!X^{^{(n)}}\!(t)\right\vert^2\!\!\leq t\!\!\int_0^t\!\!\! ds E\left\vert\Upsilon( X^{^{(n)}}\!(s), h^{^{(n)}}\!(s))\!-\! \Upsilon( X^{^{(n-1)}}\!(s), h^{^{(n-1)}}\!(s))\right\vert^2,
\end{equation*}
for $n\geq1$. Similarly,
$$
E\|h^{(n+1)}(t)-h^{(n)}(t)\|_2^2\leq 2E\int_0^t ds \left\|
\rho_{t-s}^{\gamma}[\varphi_{X^{(n)}(s)}-\varphi_{X^{(n-1)}(s)}] \right\|_2^2
$$
$$
+2t E\int_0^t ds \left\|\rho_{t-s}^{\gamma}[\Upsilon(X^{(n)}(s),h^{(n)}(s))-
\Upsilon(X^{(n-1)}(s),h^{(n-1)}(s))] \right\|_2^2\mbox{.}
$$
Being $\rho^{\gamma}_t$ a probability density, and  because $\|\rho^{\gamma}_t\varphi\|_2\leq\|\rho^{\gamma}_t\|_1\|\varphi\|_2$,
$\rho_t^{\gamma}$ is contractive on $L^2(\mathbb{R})$; therefore 
\begin{align*}
E\|(X^{(n+1)}&(t), \,h^{(n+1)}(t))-( X^{(n)}(t), h^{(n)}(t) )\|_B^2\\
&\leq 2K^2 (1+t) \int_0^t \!\!\!ds E\|(X^{(n)}(s), h^{(n)}(s))-( X^{(n-1)}(s), h^{(n-1)}(s) ) \|_B^2. 
\end{align*}
Iterating we end up with
$$
E\|(X^{(n+1)}(t), h^{(n+1)}(t))-( X^{(n)}(t), h^{(n)}(t) ) \|_B^2 \leq\frac{[2K^2(t+t^2)]^{n+1}}
{n!}\mbox{,}
$$
which gives uniform convergence on compacts $[0,T]$ of the sequence 
$\left(X^{(n)}(t),h^{(n)}(t)\right)$ to a limiting process, $\left(X(t),h(t)\right)$.  Such a process is therefore an $\mathcal{F}_t$-adapted  solution to \eqref{sistema_unicita}.\\
\textit{Uniqueness}: by what we have done so far, it is  clear that one can find a \hyphenation{suitable} suitable $c(t)$ uniformly bounded on compacts such that if $\left(\bar{X}(t),\bar{h}(t)\right)$ is another solution, then
$$
E\|(X(t),h(t))-(\bar{X}(t),\bar{h}(t))\|_B^2\leq c(t)
\int_0^t ds E\|(X(t),h(t))-(\bar{X}(t),\bar{h}(t))\|_B^2\mbox{,}
$$ 
hence uniqueness follows by the Gronwall Lemma; ($\ref{A.4}$) is then 
a consequence of continuity, which we are going to prove.\\
\textit{Continuity}: being $b(t)$ a.s. continuous and $\beta(X,h)$ bounded, $X(t)$ is a.s. continuous. In order to prove continuity for $h(t)$ we first need to prove that for any $g\in L^2(\mathbb{R})$
$$
\lim_{t\rightarrow 0}\| \rho_t^{\gamma}g-g\|_2=0\mbox{.}
$$
In fact, using the scaling property of the kernel and the Jensen inequality (weighted version), we get 
$$
\| \rho_t^{\gamma}g-g\|_2^2=
\int_\R dx \left[\int_\R dw \rho_1^{\gamma}(w) 
\left( g(x-wt^{\gamma})-g(x) \right) \right]^2
$$
$$
\leq \int_\R dx \int_\R dw \rho_1^{\gamma}(w)\,\left(g(x-wt^{\gamma})-g(x) \right)^2
$$
$$
=\int_\R dw \rho_1^{\gamma}(w)\,\| T_{wt^{\gamma}}g-g\|_2^2
$$
where $T_{\tau}$, $\tau\in\,\R$, is the translation $(T_{\tau}g)(x)=g(x-\tau)$.
Let us study the integrand:
$$
\| T_{\tau}g-g\|_2^2= C \|\widehat{T_{\tau}g}-\hat{g} \|_2^2=\int_\R d\xi\mid e^{-i\xi\tau}\,\hat{g}(\xi)-\hat{g}(\xi)\mid^2
$$
$
\Rightarrow\lim_{t\rightarrow 0}\| T_{wt^{\gamma}}g-g\|_2^2=0  
$ for a.e. $w$ and  
$$
\rho_1^{\gamma}(w)\,\| T_{wt^{\gamma}}g-g\|_2^2 \leq C \rho_1^{\gamma}(w)\|g\|_2^2\in L^1(\R),
$$
so we can apply the dominated convergence theorem and conclude.\\
We are left with the continuity of  
$k(t):=h(t)-\rho_t^{\gamma} h_0$. 
\begin{eqnarray*}
-k(t+\delta)+k(t)&=&\int_0^t db(s)\,(\rho^{\gamma}_{t+\delta-s}-\rho^{\gamma}_{t-s})\varphi_{X(s)}\nonumber\\
&+& \int_t^{t+\delta} db(s)\,\rho^{\gamma}_{t+\delta-s}\varphi_{X(s)}\nonumber\\
&+& \int_0^t ds\,\Upsilon(X(s),h(s))\,(\rho^{\gamma}_{t+\delta-s}-\rho^{\gamma}_{t-s})\varphi_{X(s)}\nonumber\\
&+& \int_t^{t+\delta} ds\,\Upsilon(X(s),h(s))\,\rho^{\gamma}_{t+\delta-s}\varphi_{X(s)}\nonumber\mbox{.}
\end{eqnarray*}
From now on we  treat the cases $0<\gamma<\frac{1}{2}$ and  $\frac{1}{2}<\gamma<1$ separately.\\ 
Let us start with the superdiffusion:
$$
E\parallel k(t+\delta)-k(t)\parallel_2^4 \leq C(A_1+A_2+A_3+A_4),
$$
where
$$
A_1:=E\left\|\int_0^t ds\,(\rho^{\gamma}_{t+\delta-s}-\rho^{\gamma}_{t-s})
\varphi_{X(s)}\right\|_2^4, 
$$
$$
A_2:=E\left\|\int_t^{t+\delta} ds\,\rho^{\gamma}_{t+\delta-s}\varphi_{X(s)}
\right\|_2^4, 
$$
\begin{equation*}
A_3:=E\left\|\int_0^t db(s)\,(\rho^{\gamma}_{t+\delta-s}-\rho^{\gamma}_{t-s})
\varphi_{X(s)}\right\|_2^4 ,
\end{equation*}
$$
A_4:=E\left\|\int_t^{t+\delta} db(s)\,\rho^{\gamma}_{t+\delta-s}
\varphi_{X(s)}\right\|_2^4.
$$
\nopagebreak[0]
We need to estimate all the above terms:
\begin{align*}
A_1&\leq C E\left[ \int_0^t ds \|(\rho^{\gamma}_{t+\delta-s}-\rho^{\gamma}_{t-s})
\varphi_{X(s)}\|_2\right]^4\nonumber\\
&=C E\left[\int_0^t ds  \|(\rho^{\gamma}_{s+\delta}-\rho^{\gamma}_{s})\varphi\|_2 \right]^4\nonumber\\
&= CE\!\left[ \int_0^t\!\! ds\!\left( \!\int_\R \!dx\!\left(\int_\R \!dz \rho_1^{\gamma}(z)
[\varphi(x-z(s+\delta)^{\gamma})- \varphi(x-z s^{\gamma})]\right)^2   
\right)^{\frac{1}{2}}\! \right]^4\nonumber\\ 
&\leq C E\left[\int_0^t ds \int_\R dz  \rho_1^{\gamma}(z)\|\varphi_{z(s+\delta)^{\gamma}}-
\varphi_{z s^{\gamma}} \|_2 \right]^4\nonumber\\ 
&\leq C E\left[\int_0^t ds \int_\R dz\rho_1^{\gamma}(z)\mid z\mid \delta^{\gamma} \right]^4 \leq C t^4\delta^{4\gamma},\nonumber
\end{align*}
having used the scaling property (\ref{scaling_property}) and ($\ref{varphib-varphia}$).
\begin{align*}
A_2&\leq E\left[ \int_\R dx \delta\int_t^{t+\delta}(\rho^{\gamma}_{t+\delta-s}\varphi_{X(s)})^2 ds   \right]^2\nonumber\\
&=\delta^2 E\left( \int_0^{\delta}ds \| \rho_s^{\gamma}\varphi_{X(t+\delta-s)}\|_2^2\right)^2\leq C\delta^4,\nonumber
\end{align*}
having used the Cauchy-Schwartz inequality and the contractivity. \\
In order to find estimates on the last two terms, let us choose $\psi(x)=\sqrt{1+\mid x\mid}$ so that 
$\forall f \in L^2(\R) \mbox{ , }  
\|f\|_2^4\leq \|\psi^{-2}\|_2^2 \|f\psi\|_4^4\;$. Hence, via the Burkholder inequality  and again Cauchy-Schwartz, we get
$$
A_3\leq \|\psi^{-2}\|_2^2 \,E\left\| \int_0^t db(s) \psi (\rho^{\gamma}_{t+\delta-s}-
\rho^{\gamma}_{t-s})\varphi_{X(s)} \right\|_4^4
$$
$$
\leq C E\,\left\|\int_0^t ds \left[ \psi (\rho^{\gamma}_{t+\delta-s}-
\rho^{\gamma}_{t-s})\varphi_{X(s)}\right]^2 \right\|_2^2
$$
$$
\leq C t \int_0^t ds E\int_\R dx \,\psi(x+X(s))^4 \left[(\rho^{\gamma}_{t+\delta-s}-\rho^{\gamma}_{t-s})\varphi\right]^4(x)
$$
\begin{equation*}
\leq Ct (1+E \sup_{u\in [0,t]} \mid X(u)\mid^2) \int_0^t ds \|\psi (\rho^{\gamma}_{s+\delta}-\rho\esp_s)\varphi \|_4^4,
\end{equation*}
having used $\psi(x+X)^4\leq (1+\mid X\mid^2)\psi^4 (x)$. Let us look at the integrand: since $\psi(x)\leq \psi(y)+\sqrt{\mid x-y\mid}$, we have
\begin{align}
\!\!\!\!\!\!\!\!\|\psi (\rho^{\gamma}_{s+\delta}-\rho\esp_s)\varphi \|_4^4\leq&
C\| (\rho^{\gamma}_{s+\delta}-\rho\esp_s)(\psi\varphi) \|_4^4\\
+&C\int_\R\!\! dx \left[ \int_\R\!\! dy  (\rho^{\gamma}_{s+\delta}(x-y)\!-\rho^{\gamma}_s(x-y))\sqrt{\mid x-y\mid} \varphi(y) \right]^4\label{triangle}.
\end{align}
The first addend can be estimated similarly to what we have done for $A_1$, so we get
$$
\| (\rho^{\gamma}_{s+\delta}-\rho\esp_s)(\psi\varphi) \|_4^4\leq C\delta^{4\gamma};
$$ 
for the second, after applying Cauchy-Schwartz on the integrand,  we find
\begin{eqnarray*}
(\ref{triangle})\!\!\!&\leq&\!\!\!\!\!C\int_\R dx\left\{ \left(\int_\R dy(\rho^{\gamma}_{s+\delta}-\rho^{\gamma}_s)(x-y)\,\mid x-y\mid\right )^2 
\left(\int_\R dy(\rho^{\gamma}_{s+\delta}-\rho^{\gamma}_s)(x-y)\varphi^2(y)\right)^2\right\}\nonumber\\
&\leq&\!\!\!\!\!C\left(\int_\R dz \rho_1^{\gamma}(z)\mid z\mid((s+\delta)^{\gamma}-s^{\gamma}) \right)^2\,\| (\rho^{\gamma}_{s+\delta}-\rho^{\gamma}_s)\varphi^2\|_2^2\leq C\delta^{4\gamma},\nonumber
\end{eqnarray*}
and we end up with
$$
A_3\leq Ct^2 \left( 1+E\sup_{u\in[0,T]}\mid X(u)\mid^2\right)\delta^{4\gamma}\mbox{.}
$$
For $A_4$, analogously, 
\begin{align*}
A_4&\leq C\delta \left( 1+E\sup_{u\in[0,T]}\mid X(u)\mid^2\right)\int_0^{\delta} ds \|\psi\rho_s\esp\varphi \|_4^4\\
&\leq C\delta \left( 1+E\sup_{u\in[0,T]}\mid X(u)\mid^2\right)\\
&\qquad\times\int_0^{\delta}ds\left\{ \|\psi\varphi\|_4^4+\int_\R\, dx\left(\int_\R\, dy \rho_s^{\gamma}(x-y)\sqrt{\mid x-y\mid}\varphi(y) \right)^4 \right\}\mbox{.}
\end{align*}
Now the integral on the second line is estimated from above by
$$
\int_0^{\delta} \!\!ds \left\{\|\psi\varphi\|_4^4 +
\left(\int_\R dz \rho_s^{\gamma}(z)
\mid \!z\!\mid \right)^2 \!\|\rho_s^{\gamma}\varphi^2 \|_2^2  \right\},
$$
so that
$$
A_4\leq C\delta^2 \left( 1+E\sup_{u\in[0,T]}\mid X(u)\mid^2
\right)\mbox{.}
$$
Proving continuity in the subdiffusive case is slightly more delicate; let us write 
$$
E\|k(t+\delta)-k(t)\|_2^{2N}\leq C (\mathcal{A}_1+\mathcal{A} _2 + \mathcal{A}_3 +\mathcal{A} _4),
$$
where $N=N(\gamma)$ is to be specified in the following and 
$$
\mathcal{A}_1:=E\left\|\int_0^t ds\,(\rho^{\gamma}_{t+\delta-s}-
\rho^{\gamma}_{t-s})\varphi_{X(s)}\right\|_2^{2N}, 
$$
$$
\mathcal{A}_2:=E\left\|\int_t^{t+\delta} ds\,\rho^{\gamma}_{t+\delta-s}
\varphi_{X(s)}\right\|_2^{2N}, 
$$
$$
\mathcal{A}_3:=E\left\|\int_0^t db(s)\,(\rho^{\gamma}_{t+\delta-s}-
\rho_{t-s})\varphi_{X(s)}\right\|_2^{2N}, 
$$
$$
\mathcal{A}_4:=E\left\|\int_t^{t+\delta} db(s)\,\rho^{\gamma}_{t+\delta-s}
\varphi_{X(s)}\right\|_2^{2N} \mbox{.}
$$
For $\mathcal{A}_2$:
$$
\mathcal{A}_2 \leq C \delta^{2N}\left\vert \int_0^{\delta} \|\rho^{\gamma}_s \varphi\|_2^2  \right\vert^{2N}\leq C \, \delta^{4N},
$$
so that we need $N>\frac{1}{4}$.\\
For $\mathcal{A}_3$: let us choose again $\psi(x)=\sqrt{1+\mid x\mid}$ as an auxiliary function; then
$\forall N>0,\,\, \|\psi^{-2}\|_{\frac{N}{N-1}}^{N}<\infty$ and
$
\forall f \in L^2(\R) \mbox{ , }  \; \|f\|_2^{2N}\leq \|\psi^{-2}\|_{\frac{N}{N-1}}^{N} \,\|f\psi\|_{2N}^{2N}\;
$.
Via the Burkholder inequality, using  $\psi^{2N}(x+X)\leq C \,(1+\mid X(u)\mid^{N})\psi^{2N}(x)$ and working as we did for $A_3$ we get
\begin{align}\label{cui_ricondursi}
\mathcal{A}_3&\leq C E\left\|\int_0^tds \,\psi^2 \,[(\rho^{\gamma}_{t+\delta-s}-
\rho^{\gamma}_{t-s})\varphi_{X(s)}]^2 \right\|_N^N\\
&\leq C \,t^{N-1} (1+E \sup_{u\in [0,t]} \mid X(u)\mid^{N}) \int_0^t\!\!ds\int_\R dx \,\psi^{2N}(x)\left\vert(\rho^{\gamma}_{s+\delta}-\rho^{\gamma}_{s})\varphi\right\vert^{2N}\!(x)\nonumber\\
&\leq C \,t^{N-1} (1+E \sup_{u\in [0,t]} \mid X(u)\mid^{N})\int_0^tds\int_\R dx \left\vert \psi(x)\, \int_s^{s+\delta}d\tau \rho'^{\gamma}_{\tau}\,\varphi\right\vert^{2N}\nonumber\\
&\leq C \,t^{N-1} (1+E \sup_{u\in [0,t]} \mid X(u)\mid^{N})\nonumber\\
& \qquad\times \int_0^tds \int_\R dx\,
\left\vert \psi(x)\int_s^{s+\delta}d\tau\, \frac{d}{d\tau}\int_0^{\tau} du\frac{\rho^{\gamma}_u \varphi''}{(\tau-u)^{1-2\gamma}}\right\vert^{2N}\nonumber\\
&= C \,t^{N-1} (1+E \sup_{u\in [0,t]} \mid X(u)\mid^{N})\nonumber\\
& \qquad\times \int_0^tds \int_\R dx\,\left\vert \psi(x)\left[ \int_0^{s+\delta} \frac{du\,\rho^{\gamma}_u \varphi''}{(s+\delta-u)^{1-2\gamma}} -\int_0^s \frac{du\,\rho^{\gamma}_u \varphi''}{(s-u)^{1-2\gamma}}\right]\right\vert^{2N}\nonumber\\
&\leq C \,t^{N-1} (1+E \sup_{u\in [0,t]} \mid X(u)\mid^{N}) \nonumber\\
& \qquad\times\left\{ \int_0^tds \int_\R dx\left\vert\psi(x)\int_0^s \!du \rho^{\gamma}_u \varphi''
\left(\frac{1}{(s+\delta-u)^{1-2\gamma}}-\frac{1}{(s-u)^{1-2\gamma}} \right) \right.\right\vert^{2N}\nonumber\\
& \qquad + \left.\int_0^tds \int_\R dx\,\left\vert \psi(x)\int_s^{s+\delta}du \frac{\rho^{\gamma}_u \varphi''}{(s+\delta-u)^{1-2\gamma}}\right\vert^{2N}\right\}
\nonumber\\
&\leq C \,t^{N-1} (1+E \sup_{u\in [0,t]} \mid X(u)\mid^{N})
[\mathcal{A}_{3a}+\mathcal{A}_{3b}],
\end{align}
\nopagebreak[0]
where
$$
\mathcal{A}_{3a}=\int_0^tds \int_\R dx \left\vert \int_0^s du 
\left(\frac{1}{(s+\delta-u)^{1-2\gamma}}-\frac{1}{(s-u)^{1-2\gamma}} \right)
\rho_u\esp\varphi''\psi\right\vert^{2N}
$$
$$
\quad+\!\int_0^t\!\!\!ds \!\!\int_\R \!\!dx \!\left\vert \int_0^s 
\!\!\!\!du \!\left(\!\frac{1}{(s\!+\!\delta\!-\!u)^{1-2\gamma}}\!-\!\frac{1}{(s\!-\!u)^{1-2\gamma}} \!\right)\!\!\int_\R\!\! dy\rho^{\gamma}_u (x\!-\!y)\varphi''(y)\sqrt{\mid\!\! x\!-\!y\!\!\mid}\right\vert^{2N}
$$
$$
\mathcal{A}_{3b}= \int_0^tds \int_\R dx \left\vert 
\int_s^{s+\delta} du \frac{1}{(s+\delta-u)^{1-2\gamma}}
\rho^{\gamma}_u\varphi''\psi\right\vert^{2N}
$$
$$
\quad+\int_0^tds \int_\R dx \left\vert \int_s^{s+\delta}  
du \frac{1}{(s+\delta-u)^{1-2\gamma}} \int_\R dy
\rho^{\gamma}_u (x-y)\varphi''(y)\sqrt{\mid x-y\mid}\right\vert^{2N}\mbox{.}
$$
We claim that
\begin{eqnarray}\label{max_finiti}
& &\int_\R dx \left\vert\max_{0\leq u\leq s+\delta} (\rho^{\gamma}_u
\varphi''\psi)(x)\right\vert^{2N} <\infty,\nonumber\\
& &\int_\R dx \left\vert\max_{0\leq u\leq s+\delta} (\rho^{\gamma}_u(\cdot)
\sqrt{\cdot}\ast\varphi'')(x)\right\vert^{2N} <\infty\nonumber\mbox{.}
\end{eqnarray}
Indeed, 
$
\rho^{\gamma}_u\varphi''\psi
$
is continuous in $u$, so the maximum in ($\ref{max_finiti}$)$_1$ is attained at, say, $\tilde{u}$ and  
$
\|\rho^{\gamma}_{\tilde{u}}\varphi''\psi\|_{2N}^{2N}\leq \,C
$. 
The maximum in ($\ref{max_finiti}$)$_2$ is reached at the second extremum $(s+\delta)$, in fact
$$
\int_\R dx \left( \int_\R dy \rho^{\gamma}_{u}(x-y)\sqrt{\mid x-y\mid}\,\varphi''(y)\right)^{2N}
$$
$$
\leq C \int_\R dx \left(\int_\R dz \rho^{\gamma}_{1}(z)\mid z\mid^N u^{\gamma}\right)^{2N}\mbox{.}
$$
Therefore,
$$
\mathcal{A}_{3a}\leq C \int_0^t ds \left\vert \int_0^s du \left(\frac{1}{(s+\delta-u)^{1-2\gamma}}-\frac{1}{(s-u)^{1-2\gamma}} \right)\right\vert^{2N}=C(t) \delta^{(1-2\gamma)2N},
$$
with $C(t)$ bounded on compacts and 
$$
\mathcal{A}_{3b}\leq C \int_0^t ds\left\vert \int_s^{s+\delta} du 
\frac{1}{(s+\delta-u)^{1-2\gamma}} \right\vert^{2N}=C t \delta^{4N\gamma}.
$$
In order to apply the Kolmogorov criterion we need $4N\gamma>1$ and  $(1-2\gamma)2N\!>1$.
For $\mathcal{A}_1$ and $\mathcal{A}_4$:
$$
\mathcal{A}_1 \leq C \left\|\int_0^t ds \psi(x)
(\rho^{\gamma}_{t+\delta-s}-\rho^{\gamma}_{t-s})\varphi_{X(s)}\right\|_{2N}^{2N}
$$
$$
\leq C t^{2N-1} E \int_\R dx \int_0^t ds \mid \psi(x)(\rho^{\gamma}_{t+\delta-s}-
\rho^{\gamma}_{t-s})\varphi_{X(s)}\mid^{2N},
$$
which is exactly (\ref{cui_ricondursi}).
$$
\mathcal{A}_4 \leq C (1+E \sup_{u\in [0,t]} \mid X(u)\mid^{2N})\,\delta^{2N-1}\,\int_0^{\delta} \|\psi\rho^{\gamma}_s \varphi\|_{2N}^{2N};
$$
with analogous calculations, the integrand on the right hand side is bounded,  
hence
$$
\mathcal{A}_4 \leq C (1+E \sup_{u\in [0,t]} \mid X(u)\mid^{2N})\,\delta^{2N}\mbox{.}
$$
To conclude, requiring
$$
\left\{
\begin{array}{ll}
N\geq \frac{1}{2(1-2\gamma)} & \mbox{ if }\; \gamma\geq\frac{1}{4}\\
N\geq \frac{1}{4\gamma} & \mbox{ if }\; \gamma\leq\frac{1}{4},
\end{array}
\right.
$$
continuity follows.
\end{proof}
\section{Motivation}\label{app:motivation}
In the introduction we have briefly discussed the choice of the operators of fractional differentiation and of the fractional Laplacian. In this Appendix, we want to show how the operators $D_t^{\gamma}$ and $I_t^{\gamma}$ naturally arise in the context of anomalous diffusion and explain in some more detail the link with CTRWs.\\
We want to determine
an operator $A$ s.t.
$$
\left\{
\begin{array}{l}
\partial_t\rho^{\gamma}_t(x)=A\,\rho^{\gamma}_t(x)\\
\rho^{\gamma}_t(0)=\delta_0,
\end{array}
\right.
$$
with $\rho\esp(t,x)$ enjoying the following three properties:
\begin{equation}\label{richieste_su_rho}
\int_{\R} dx \rho^{\gamma}_t(x)=1\,\mbox{,}\quad
\int_{\R} dx \rho^{\gamma}_t(x) \,x = 0 \quad \mbox{e}\quad
\int_{\R} dx \rho^{\gamma}_t(x) \,x^2 \sim t^{2\gamma}\,
\end{equation}
(notice that for $\gamma=\frac{1}{2}$ we recover the diffusion equation with $A=\Delta$). We recall that $\hat{f}$, $f^{\#}$ and $\tilde{f}$ denote  the Fourier, the Laplace and the Fourier-Laplace transform of the function $f$, respectively.\\
By  ($\ref{richieste_su_rho}$), the following must hold
$$
\hat{\rho}^{\gamma}_t(k)=1-\frac{1}{2}ct^{2\gamma}k^2+o(k^2)\qquad\mbox{and}
$$
$$
\tilde{\rho}^{\gamma}(\mu,k)=\frac{1}{\mu}-\frac{ck^2}{2\mu^{2\gamma+1}}\Gamma(2\gamma+1)=
\frac{1}{\mu}(1-c_1\mu^{-2\gamma}k^2)\mbox{,}
$$
where $c_1=\frac{1}{2}c\Gamma(2\gamma+1)$.
In definitions (\ref{eqn_soddisf_da_rho_sub_frazionaria}) and 
(\ref{eqn_soddisf_da_rho_super_frazionaria}) the constant $c_1$ should appear; we just set it equal to $1$ both for simplicity and not being interested,  in this context, in estimating the "anomalous diffusion" constant.\\ 
We can assume that the expression for $\tilde{\rho}^{\gamma}(\mu,k)$ is valid in the regime $\mu^{-2\gamma}k^2<<1$. Actually, condition
($\ref{richieste_su_rho}$)$_3$ is meant for an infinitely wide system and for long times. In other words, if $\Lambda$ is the region where the particle moves, we claim that
$$
\lim_{t\rightarrow\infty}\lim_{\Lambda\rightarrow \R}\frac{\int_{\Lambda}dx \rho^{\gamma}_t(x) \,x^2 }{t^{2\gamma}}=\mbox{const}\mbox{.}
$$
This means that we are interested in the case  $k<<\mu$. Of course one can in principle find an infinite number of functions s.t. $\tilde{\rho}^{\gamma}(\mu,k)=\frac{1}{\mu}(1-c_1\epsilon)$ for $\epsilon=\mu^{-2\gamma}k^2$. One possible choice is
\begin{equation}\label{trasf_di_Four-Laplace_di_rho}
\tilde{\rho}^{\gamma}(\mu,k)=\frac{1}{\mu(1+c_1\epsilon)}=\mu^{\gamma-1}\frac{\mu^{\gamma}}{\mu^{2\gamma}+(c_1 k)^2}=\frac{1}{\mu+c_1k^2\mu^{1-2\gamma}}\mbox{,}
\end{equation}
which leads to an integro-differential equation and, when $\gamma=\frac{1}{2}$, it coincides with the Fourier-Laplace transform of a Gaussian density.\\
We now find the operator whose fundamental solution is $\tilde{\rho}^{\gamma}(\mu,k)$. We have
$$
\mathcal{L}(\partial_t\hat{\rho}^{\gamma}(\cdot,k))(\mu)=-1+\mu\tilde{\rho}^{\gamma}(\mu,k)=-c_1k^2\mu^{1-2\gamma}
\tilde{\rho}^{\gamma}(\mu,k)\mbox{.}
$$
Let $p=2\gamma-1$ and $\phi_p(t)=\frac{t^{p-1}}{\Gamma(p)}$; then we need to distinguish two cases in order to study the right hand side of the above equation:\\
when $0<\gamma<\frac{1}{2}$ one can easily check that
$$
\mathcal{L}(\phi_p\ast\hat{\rho}^{\gamma}(k,\cdot))=\tilde{\rho}^{\gamma}(\mu,k)\mu^{-p}
$$
which implies that
$$
\tilde{\rho}^{\gamma}(\mu,k)\mu^{1-2\gamma}\;\;\;\mbox{is the Laplace transform of}\;\;\; \frac{1}{\Gamma(2\gamma-1)}\int_0^tds\frac{\hat{\rho}^{\gamma}(s,k)}{(t-s)^{2-2\gamma}}\mbox{;}
$$
when $\frac{1}{2}<\gamma<1$, instead, a straightforward calculation shows that
$$
\mathcal{L}[\partial_t(\phi_{p+1}\ast\hat{\rho}^{\gamma}(k,\cdot))]=\tilde{\rho}^{\gamma}(\mu,k)\mu^{-p}
$$
so that
$$
\tilde{\rho}^{\gamma}(\mu,k)\mu^{1-2\gamma}\;\;\;\mbox{is the Laplace transform of}\;\;\; \frac{1}{\Gamma(2\gamma)}\frac{d}{dt}\int_0^tds\frac{\hat{\rho}^{\gamma}(s,k)}
{(t-s)^{1-2\gamma}}\mbox{.}
$$
Finally, taking the inverse Fourier transform, we get that $\rho^{\gamma}(t,x)$ satisfies ($\ref{eqn_soddisf_da_rho_sub_frazionaria}$) when
$0<\gamma<\frac{1}{2}$ and ($\ref{eqn_soddisf_da_rho_super_frazionaria}$) when $\frac{1}{2}<\gamma<1$.
Moreover, the explicit expression for $\rho^{\gamma}_t(x)$ holds true: by ($\ref{trasf_di_Four-Laplace_di_rho}$) we get that
$$
\tilde{\rho}^{\gamma}(\mu,k)=
\int_R dx \,e^{ikx}\,\frac{\mu^{\gamma-1}}{2\sqrt{c_1}}\,e^{-\frac{\mu^{\gamma}}{\sqrt{c_1}}\mid x\mid}
$$
hence
$$
\rho^{\#}(x,\mu)=\frac{\mu^{\gamma-1}}{2\sqrt{c_1}}\,e^{-\frac{\mu^{\gamma}}{\sqrt{c_1}}\mid x\mid}
$$
and now, by the inverse Laplace formula, we obtain ($\ref{espressione_esplicita_per_rho}$). Obviously, the expression ($\ref{espressione_esplicita_per_rho}$) has been deduced after having chosen ($\ref{trasf_di_Four-Laplace_di_rho}$) among all possible candidates for $\tilde{\rho}^{\gamma}$ and this choice can now be justified in view of the link with CTRWs.

\end{document}